\documentclass{article}
\usepackage{fullpage}
\usepackage{amsmath}
\usepackage{amsthm}
\usepackage{amssymb}
\usepackage{url}
\usepackage{graphicx}
\usepackage{verbatim}
\usepackage{listings}
\usepackage{url}
\usepackage{caption}
\usepackage{subcaption}
\newcommand{\argmax}{\operatornamewithlimits{argmax}}
\usepackage{tikz}
\usepackage[round]{natbib}
\usetikzlibrary{calc}
\usetikzlibrary{patterns}
\tikzstyle{mybox} = [draw=black, fill=white,  thick,
    rectangle, inner sep=10pt, inner ysep=20pt]
\tikzstyle{mybox} = [draw=black, fill=white,  thick,
    rectangle, inner sep=2pt, inner ysep=2pt]

\newtheorem{thm}{Theorem}
\newtheorem{lemma}{Lemma}
\newtheorem{cor}{Corollary}
\newtheorem{prop}{Proposition}
\theoremstyle{definition}
\newtheorem{definition}{Definition}
\newtheorem{remark}{Remark}
\newtheorem{example}{Example}

\begin{document}


\title{Robust Vertex Enumeration for Convex Hulls in High Dimensions\footnote{A conference version of the article will appear in the Proceedings of AISTATS 2018}}
\date{}
\author{Pranjal Awasthi \\ {\tt pranjal.awasthi@cs.rutgers.edu} \and  Bahman Kalantari \\ {\tt kalantar@cs.rutgers.edu} \and Yikai Zhang \\ {\tt yz422@cs.rutgers.edu}
}

\maketitle

\begin{abstract}
The problem of computing the vertices of the convex hull of a given finite set of  points in the Euclidean space is a classic and fundamental problem, studied in the context of computational geometry, linear and convex programming,  machine learning and more.  In this article we present {\it All Vertex Triangle Algorithm} (AVTA),  a robust and efficient algorithm that for a given input set $S= \{v_i \in \mathbb{R} ^m: i=1, \dots, n\}$  computes the subset $\overline S$ of all $K$ vertices of the convex hull of $S$. If desired AVTA computes an approximation to $\overline S$ and it can also work if the input data is a perturbation of $S$.  Let $R$ be the diameter of $S$. We say $conv(S)$, the convex hull of $S$, is $\Gamma_*$-robust if the minimum of the distances from each vertex to the convex hull of the remaining vertices is $\Gamma_*$.  Given $\gamma \leq \gamma_* = \Gamma_*/R$, the number of operations of $AVTA$ to compute $\overline S$ is $O(nK(m+  \gamma^{-2}))$. Even without the knowledge of $\gamma_*$, but when $K$ is known, using binary search, the complexity of AVTA is $O(nK(m+  \gamma_*^{-2})) \log(\gamma_*^{-1})$.   More generally, without the knowledge of $\gamma_*$ or $K$, given any $t \in (0,1)$, AVTA computes a subset $\overline S^t$ of $\overline S$ of cardinality $K^{(t)}$ in $O(n K^{(t)}(m+  t^{-2}))$ operations so that the Euclidean distance between any point $p \in conv(S)$ to $conv(\overline S^t)$ is  at most  $t R$.

Next we consider AVTA under perturbation since in practice the input maybe a perturbation of $S$, $S_\varepsilon =\{v^{\varepsilon}_i: i=1, \dots, n\}$, where $\Vert v_i -  v^{\varepsilon}_i \Vert \leq \varepsilon R$.  The set of perturbed vertices, $\overline S_\varepsilon$ may differ drastically from the set of vertices of $conv(S_\varepsilon)$. Let $\Sigma_*$ be the minimum of distances of vertices of $conv(S)$ to the convex hull of the remaining point of $S$.
Under the assumption that $\sigma_* = \Sigma_*/R \geq 4 \varepsilon$,  given $\sigma$ satisfying $4 \varepsilon \leq \sigma \leq \sigma_*$, AVTA computes $\overline S_\varepsilon$  in $O(nK_\varepsilon(m+  \sigma^{-2}))$,  where $K \leq K_\varepsilon \leq n$.
 When only $K$ is known, but assuming $ 4\varepsilon \leq \sigma_*$, using binary search the complexity of AVTA to compute $\overline S_\varepsilon$ is $O(nK(m+  \sigma_*^{-2})) \log(\sigma_*^{-1})$. More generally, given any $t \in (0,1)$, AVTA computes  a subset $\overline S_\varepsilon^t$ of $\overline S_\varepsilon$ of cardinality $K^{(t)}_\varepsilon$ in $O(n K^{(t)}_\varepsilon (m+ t^{-2})$  so that
the distance between any point $p \in conv(S)$ to $conv(\overline S_\varepsilon^t)$ is at most  $(t+\varepsilon) R$.

We also consider  the application of AVTA in the recovery of vertices through the  projection of $S$ or $S_\varepsilon$ under  a Johnson-Lindenstrauss randomized linear projection $L : \mathbb{R}^{m} \rightarrow  \mathbb{R}^{m'}$. Denoting $U=L(S)$  and $U_\varepsilon=L(S_\varepsilon)$,
by relating the robustness parameters of $conv(U)$ and $conv(U_\varepsilon)$ to those of $conv(S)$ and $conv(S_\varepsilon)$, we derive analogous complexity bounds for probabilistic computation of the vertex set of $conv(U)$ or those of $conv(U_\varepsilon)$, or an approximation to them. Finally, we apply AVTA to design new practical algorithms for two popular machine learning problems: topic modeling and non-negative matrix factorization. For topic models, our new algorithm leads to significantly better reconstruction of the topic-word matrix than state of the art approaches ~\citet{arora2013practical, bansal2014provable}. Additionally, we provide a robust analysis of AVTA and empirically demonstrate that it can handle larger amounts of noise than existing methods. For non-negative matrix we show that AVTA is competitive with existing methods that are specialized for this task~\citet{arora2012computing}.

\end{abstract}



\section{Introduction}
In this article we present {\it All Vertex Triangle Algorithm} (AVTA),  a robust and efficient algorithm that for given input set $S= \{v_1, \dots, v_n\} \subset \mathbb{R} ^m$, computes the subset $\overline S=\{\overline v_1, \dots, \overline v_K\}$ of all vertices of $conv(S)$, the convex hull of $S$. More generally, given any $t \in (0,1)$, AVTA computes a subset $\overline S^t$ of $\overline S$ so that
the distance between any point $p \in conv(S)$ to $conv(\overline S^t)$ is to within a distance of $t R$. AVTA is also applicable if the input date is a perturbation of $S$.

AVTA, a {\it fully polynomial-time approximation scheme}, builds upon the  {\it Triangle Algorithm}   \citet{kalantari2015characterization}, designed to solve the {\it the convex hull membership problem}.  Specifically, given $S$, the Triangle Algorithm tests if a distinguished point $p$ lies in the $conv(S)$,  either by computing a point  $p_\varepsilon \in conv(S)$ to within a prescribed distance to $p$, or a hyperplane that separates  $p$ from $conv(S)$.  Before describing AVTA and its applications we wish to give an overview of the related problems and research, as well as their history, significance and connections to our work.

The convex hull membership problem is a basic problem in computational geometry and a very special case of the {\it convex hull problem}, see  Goodman and  ~\citet{toth2004handbook}. Besides being a fundamental problem in computational geometry, it is a basic problem in linear programming (LP). In fact LP over the integers can be reduced to a convex hull membership problem.  Furthermore, the two most famous polynomial-time LP algorithms, the ellipsoid algorithm of  ~\citet{khachiyan1980polynomial} and the projective algorithm of ~\citet{karmarkar1984new}, are in fact explicitly or implicitly designed to  solve the convex hull membership problem when $p=0$, see ~\citet{jin2006procedure}. Furthermore, using an approach suggested by  Chv\'atal,  in ~\citet{jin2006procedure}it can be shown that there is a direct connection between a general LP feasibility and this homogeneous case of the convex hull membership problem.

An important problem  in computational geometry and machine learning is the {\it irredundancy problem}, the problem of computing all the vertices of  $conv(S)$, see ~\citet{toth2004handbook}.  Clearly, any algorithm for LP feasibility can be used to solve the irredundancy problem by solving a sequence of $O(n)$ convex hull membership problems.  For results that reduce the number of linear programming problems,  see e.g.  ~\citet{clarkson1994more} and  ~\citet{chan1996output}.  Some applications require the description of  $conv(S)$ in terms of its vertices, facets and adjacencies,  see  ~\citet{chazelle1993optimal}.  The complexity of  many exact algorithms for irredundancy is exponential in terms of the dimension of the points, thus only  practical in very low dimensions.   On the other hand, the convex hull membership problem by itself has been studied in the context of large scale applications where simplex method or polynomial time algorithms are too expensive to run.  Thus approximation schemes have been studied for the problem.

~\citet{blum2016sparse}  propose a bi-criterion algorithm based on \textit{Nearest Neighbot Oracle},  computing a subset of vertices $T$ satisfying two properties: i) the Hausdorff distance between $conv( T )$ and $conv(S)$ is bounded above by $(8 \varepsilon^{1/3} + \varepsilon) R$
(ii) $|T| =O(K_{opt}/\varepsilon^{2/3})$.
Since $ T \subset S$, this implies that
$\varepsilon = \Omega((K_{opt}/n)^{3/2})$.  The running time of the algorithm is
$$O\bigg (\frac{nK_{opt}}{\varepsilon^{2/3}} \bigg (m
+ \frac{K_{opt}}{\varepsilon^{8/3}} + \frac{K^2_{opt}}{\varepsilon^{4/3}} \bigg ) \bigg).$$
While there is a theoretical bound on the size of $T$ as a polynomial in $1/\varepsilon$,  it is in-efficient  since it uses the  \textit{Nearest Neighbot Oracle}. Indeed, in AVTA, the Triangle algorithm works as an approximate oracle which achives great improvement in efficiency.
Given that $\bar{S}$ is $\gamma$ robust and additionally  $\gamma$ is $\Omega(\varepsilon^{1/3})$, then we cannot use fewer than $|\bar{S}|$ vertices to give an $\varepsilon$ approximation. This argument shows that in ~\citet{blum2016sparse} $K_{opt}=|\bar{S}|$. In a general case where $\gamma $ is arbitrarily close to $0$, AVTA will find all vertices in  $O(nK_{\varepsilon}(m+\frac{1}{\varepsilon^2}))$ time. While we so far have no nontrivial bound on $K_{\epsilon}$, it is known that $K_{\varepsilon}\leq n$. In this case the complexity of AVTA is $O(n^2m+ n^2/\varepsilon^{2})$ and Greedy clustering requires at least $O(nm/\varepsilon^{2}+ n/\varepsilon^{10})$ to achieve the same accuracy. It could be concluded that there exists regimes that AVTA outperforms \textit{Greedy Clustering}. It is interesting to observe that AVTA could be used as a pre-processing algorithm for
Greedy Clustering. By our analysis, AVTA only detects vertices and will not omit any of them. In case $n >> K_{\epsilon}$, we can use AVTA to delete points inside the convex hull thus reduce the size of the problem for Greedy Clustering. In summary, the two algorithms coexist.

Not only is convex hull detection a fundamental problem in computational geometry, state of the art algorithms for many machine learning problems rely on being able to solve this problem efficiently. Consider for instance the
problem of non-negative matrix factorization~(NMF)~\citet{lee2001algorithms}. Here, given access to a data matrix $A$, we want to compute non-negative, low rank matrices $U$ and $V$ such that $A=UV$. Although in general this problem is intractable, recent results show that under
a natural {\em separability} assumption~\citet{donoho2003does} such a factorization can be computed efficiently~\citet{arora2012computing}. The key insight in these works is that under the separability
assumption, the rows of the matrix $V$ will appear among the rows of $A$. Furthermore, the rows of $V$ will be the {\em vertices} of the convex hull of rows of $A$. Hence, a fast algorithm for detecting the vertices will lead to a fast factorization algorithm as well.

A problem related to NMF is known as {\em topic modeling}~\citet{blei2012probabilistic}.  Here one is given access to a large corpus of documents, with each document represented as a long vector consisting of frequency in the document of every word in the vocabulary. This is known as the bag-of-words representation. Each document is assumed to represent a mixture of up to $K$ hidden topics. A popular generative model for such documents is the following: For every document $d$, a $K$ dimensional vector $\theta_d$ is drawn from a distribution over the simplex. Typically this distribution is the Dirichlet distribution. Then, for each word in the document, a topic is chosen according to $\theta_d$. Finally, given a chosen topic $i$, a word is output according to the topic distribution vector $\beta_i$. This is known as the Latent Dirichlet Allocation~(LDA) model~\citet{blei2003latent}. The parameters of this model consist of the topic-word matrix $\beta$ so that $\beta_i$ defines the distribution over words for topic $i$. Additionally, there are hyper parameters associated with the Dirichlet distributions generating the
topic distribution vector $\theta_d$. The topic modeling problem concerns learning the topic-word matrix $\beta$ and the parameters of the topic generating distribution. Similar to NMF, the problem is intractable in the worst case but can be efficiently solved
under {\em separability}~\citet{arora2012learning}. In this context, the separability assumption requires that for each topic $i$, there exists an {\em anchor word} that has a non-zero probability of occurring only under topic $i$. Separability is an assumption that is known to hold for real world documents~\citet{arora2012learning}. The key component towards learning
the model parameters is a fast algorithm for finding the anchor words.
The algorithm of ~\citet{arora2012learning, arora2013practical} uses the word-word covariance matrix and shows that under separability, the vertices of the convex hull of the rows of the matrix will correspond to the anchor words. Similarly, the work of~\citet{ding2013topic} shows that finding the vertices of the convex hull of the document-word matrix will also lead to detection of anchor words. Both approaches rely on the vertex detection subroutine. Furthermore, in the case of topic models, the documents are limited in size and this translates to the fact that one is given a perturbation of the set $S$. The goal is to use this perturbed set to approximate the original vertices $\overline S$. Hence in this application it is crucial that the approach to finding the vertices be robust to noise.

The convex hull membership problem can be formulated as the minimization of a convex quadratic function over the unit simplex. This particular convex program  finds  applications in statistics, approximation theory, and machine learning,  see e.g ~\citet{clarkson2010coresets} and ~\citet{zhang2003sequential} who consider the analysis of a greedy algorithm for  minimizing  smooth convex functions over the unit simplex.  The Frank-Wolfe algorithm ~\citet{frank1956algorithm} is a classic greedy algorithm for convex programming.   When the the convex hull of a set of points does not contain the origin, the problem of computing the point in the convex hull with least norm, known as {\it polytope distance} is also a problem of interest.  In some applications the  polytope distance refers to the distance  between two convex hulls, a fundamental problem in machine learning, known as SVM,  see e.g. ~\citet{burges1998tutorial}.  Gilbert's algorithm ~\citet{gilbert1966iterative} for the polytope distance problem is one of the earliest known algorithms.  ~\citet{gartner2009coresets} show Gilbert's algorithm coincides with Frank-Wolfe algorithm when applied to the minimization of a convex quadratic function over a  unit simplex. In this case the algorithm is known as {\it sparse greedy approximation}.  For many results regarding the applications of the minimization of a quadratic function over a simplex, see ~\citet{zhang2003sequential}, ~\citet{clarkson2010coresets} and ~\citet{gartner2009coresets}. ~\citet{clarkson2010coresets} analyzes the  Frank-Wolfe and its variations while studying the notion of  {\it coresets}. While the Triangle Algorithm has features that are very similar to those of Frank-Wolfe algorithm, there are other features and properties that make it an algorithm distinct from Frank-Wolfe or Gilbert's algorithm.   To describe these differences, consider the distance between $p$ and  $conv(S)$:
\begin{equation} \label{euclid}
\Delta= \min \bigg \{d(p',p) \equiv \Vert p'-p \Vert: \quad p' \in conv(S)\}=d(p_*,p) \bigg \}.
\end{equation}

Clearly, $p \not \in conv(S)$, if and only if $\Delta >0$.  The goal of the convex hull membership problems (equivalently an LP feasibility) is to test feasibility, i.e. if $p$ lies in $conv(S)$. Solving this does not require the computation of $\Delta$ when it is positive.  Thus the goal of solving the convex hull membership is different from that of computing this distance $\Delta$when positive.  When $p \in conv(S)$, the analysis of complexity of the Triangle Algorithm is essentially identical with ~\citet{clarkson2010coresets} analysis of the basic Frank-Wolfe algorithm.  G{\"a}rtner and Jaggi ~\citet{gartner2009coresets} on the other hand analyze the complexity of Gilbert's algorithm for the polytope distance problem, i.e. the approximation of $\Delta$, however under the assumption that $\Delta >0$.  ~\citet{gartner2009coresets} do not address the case when $\Delta =0$.

What  distinguishes the Triangle Algorithm from the Frank-Wolfe and Gilbert's algorithms is the {\it distance dualities} which gives more flexibility to the algorithm.  The algorithm we will analyze in this article, namely AVTA, is designed to generate all vertices of $conv(S)$. It makes repeated use of the distance dualities of the Triangle Algorithm, resulting in an over all efficient algorithm for computing the vertices of $conv(S)$, or very good approximation to these vertices, even under perturbation of the input set. Indeed AVTA is testimonial to the uniqueness of the Triangle Algorithm while itself is a nontrivial extension of the Triangle Algorithm. AVTA  finds many applications in computational geometry and machine learning. Some of these are demonstrated here theoretically and computationally.  We next describe AVTA in more detail.

To describes the complexities of AVTA we need to define some parameters.  We say $conv(S)$ is $\Gamma_*$-{\it robust}, if $\Gamma_*$ is the minimum of  the distances from each $\overline v_i \in \overline S$ to $conv(\overline S \setminus \{\overline v_i\})$. Set $R= \max \{d(v_i,v_j), v_i, v_i \in S\}$,
the diameter of $S$. AVTA works as follows.

(1) If a number  $0 < \gamma \leq \Gamma_*/R$ is known, the number of operations of $AVTA$ to  computes $\overline S$ is.
\begin{equation}O(nK(m+  \gamma^{-2})).\end{equation}

(2) If only $K$ is known, the number of operations of $AVTA$ to  compute $\overline S$ is
\begin{equation}O(nK(m+  \gamma_*^{-2}))\log \gamma_*^{-1}).\end{equation}

(3) More generally, given any $t \in (0,1)$, AVTA can compute a subset  $\overline S^t$ of $\overline S$ so that the distance of each point in $conv(S)$  to $conv(\overline S^t)$ is at most $tR$. The corresponding number of operations is
\begin{equation}O(nK^{(t)}(m+ t^{-2})), \quad K^{(t)}= |\overline S^t|.\end{equation}

In practice the input set may be not $S$ but a perturbation of it, $S_\varepsilon =\{v^{\varepsilon}_1, \dots, v^{\varepsilon}_n\}$,
where $\Vert v_i -  v^{\varepsilon}_i \Vert \leq \varepsilon R_S$.  The set of perturbed vertices, $\overline S_\varepsilon = \{ \overline v_1^{\varepsilon}, \dots, \overline v_K^{\varepsilon}\}$ may differ considerably from the set of actual vertices of $conv(S_\varepsilon)$.  Under mild assumption on $\varepsilon$, AVTA computes $\overline S_\varepsilon = \{ \overline v_1^{\varepsilon}, \dots, \overline v_K^{\varepsilon}\}$. More generally, given any $t \in (0,1)$, AVTA computes  a subset $\overline S_\varepsilon^t$ of $\overline S_\varepsilon$ so that
the distance from any  $p \in conv(S)$ to $conv(\overline S_\varepsilon^t)$ is at most $(t+\varepsilon) R$.  The complexity of AVTA for this variation of the problem is analogous to the unperturbed case, however it makes use a weaker parameter.  We say $conv(S)$ is $\Sigma_*$-{weakly robust},  if  $\Sigma_*$
is the minimum of the distances of each vertex in $S$ to the convex hull of  all the remaining points in $S$.
 In Figure \ref{Figsigma} we show a simple example where $\Gamma_*$ and $\Sigma_*$ are shown for set of eight points.

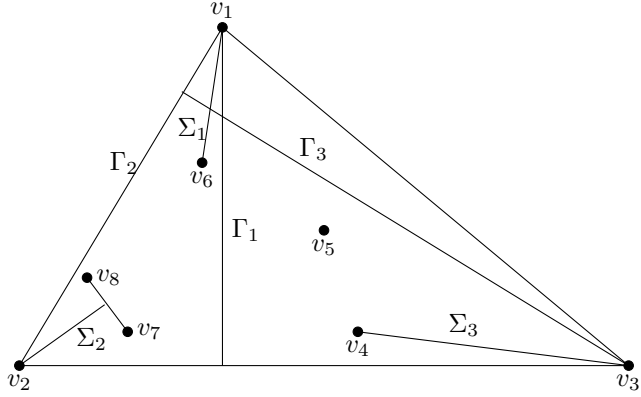
\begin{figure}[htpb]
	\centering
	\begin{tikzpicture}[scale=0.9]

		\draw (-2.0,0.0) -- (7.0,0.0) -- (1.0,5.0) -- cycle;
		\filldraw (1,5) circle (2pt) node[above] {$v_1$};
		\filldraw (-2,0) circle (2pt) node[below] {$v_2$};
		\filldraw (7,0) circle (2pt) node[below] {$v_3$};
\filldraw (3,.5) circle (2pt) node[below] {$v_4$};
\draw (7,0) --(3,.5);
 \draw (4.2,.6) node[right] {$\Sigma_3$};
 \draw (1,5) --(.7,3);
 \draw (.93,3.5) node[left] {$\Sigma_1$};
\filldraw (2.5,2) circle (2pt) node[below] {$v_5$};
\filldraw (.7,3) circle (2pt) node[below] {$v_6$};
\filldraw (-.4,.5) circle (2pt) node[right] {$v_7$};
\filldraw (-1,1.3) circle (2pt) node[right] {$v_8$};
\draw (-.4,.5) -- (-1,1.3);
\draw (-2,0) -- (-.74,.9);
 \draw (-1.3, .4) node[right] {$\Sigma_2$};

		\draw (1,5) -- (1,0.0);
        \draw (7,0) -- (.4,4.05);
        \draw (1, 2) node[right] {$\Gamma_1$};
         \draw (-.8, 3) node[right] {$\Gamma_2$};
          \draw (2, 3.2) node[right] {$\Gamma_3$};
	\end{tikzpicture}
	
\begin{center}
\caption{$\Gamma_*=\Gamma_1$ and $\Sigma_* =\Sigma_2$.} \label{Figsigma}
\end{center}
\end{figure}

We first prove when $\sigma_*=\Sigma_*/R \geq 4 \varepsilon$,  $\overline S_\varepsilon$ is a subset of vertices of $conv(S_\varepsilon)$ and $conv(S_\varepsilon)$ is at least  $\Sigma_*/2$-weakly robust. Using this, we prove

(i) If $\sigma \leq \sigma_*=\Sigma_*/R$ is known to satisfying $4\varepsilon \leq \sigma$, the number of operations of $AVTA$ to  computes $\overline S_\varepsilon$, is.
\begin{equation}O(nK_\varepsilon(m+  \sigma^{-2})),\end{equation}
where $K_\varepsilon$ is at most the cardinality of the set of vertices of $S_\varepsilon$.

Clearly $\Gamma_ * \geq \Sigma_*$, however we prove
\begin{equation}\Sigma_* \geq .5 \Gamma_* \rho_*,\end{equation}
where $\rho_*$ is the  minimum distance between distinct pair of points in $S$. This allows deriving lower bound to $\Sigma_*$ from a known lower bound on $\Gamma_*$. Thus we can alternatively write

(ii) If $\gamma \leq \gamma_*=\Gamma_*/R$ is known satisfying $4 \varepsilon \leq \gamma \rho_*/R$, the number of operations of $AVTA$ to  computes $\overline S_\varepsilon$ is.
\begin{equation}O(nK_\varepsilon(m+  (\gamma \rho_*)^{-2})).\end{equation}

(iii) If only $K$ is known, where $4 \varepsilon \leq \sigma_*=\Sigma_*/R$, the number of operations of $AVTA$ to  computes $\overline S_\varepsilon$ is.
\begin{equation}O(nK_\varepsilon(m+  \sigma_*^{-2})) \log (\sigma_*^{-1}).\end{equation}

(iv) More generally, given any $t \in (0,1)$, AVTA can compute a  subset  $\overline S_\varepsilon^t$ of $\overline S_\varepsilon$ so that the distance from each $p$ in $conv(S)$  to $conv(\overline S_\varepsilon^t)$ is at most $(t+ \varepsilon) R$. The corresponding number of operations is

\begin{equation}O(nK_\varepsilon^t(m+ t^{-2})), \quad K_\varepsilon^t= |  \overline S_\varepsilon^t|.\end{equation}

\begin{table}[!t]
	\renewcommand{\arraystretch}{1.0}
	
	\centering
	\scalebox{0.7}
	{
		\begin{tabular}{|l|l|l|c|}
			
			\hline
			~~~Input and Description &  ~~~~~~~Computed via AVTA & ~~~~~~~~~~~Conditions  & Complexity
			\\
			\hline
			& & &
			\\
			$S=\{v_1, \dots, v_n \} \subset \mathbb{R}^m$  & $\overline S$, vertices of $conv(S)$, $|\overline S|=K$ & $\gamma  \leq \gamma_* \equiv \Gamma_*/R$ is known & $O \big (nK(m+\gamma^{-2}) \big )$
			\\
			$R=\max \{\Vert v_i - v_j\Vert: v_i, v_j \in S\}$& $\overline S= \{\overline v_1, \dots, \overline v_K\}$ & Only $K$ is known  & $O \big (nK(m+ \gamma_*^{-2}) \big )
			\times \log (\gamma_*^{-1})$
			\\
			& Given $t  \in (0,1)$, $\overline S^t \subset \overline S$, $|\overline S^t|=K^{(t)}$ & General Case & $O \big (nK^{(t)}(m+t^{-2}) \big )$
			\\
			&$d(p, conv(\overline S^t) \leq t R, \forall p \in conv(S)$ & &
			\\
			\hline
			& & &
			\\
			$S_\varepsilon=\{v^{\varepsilon}_1, \dots, v^{\varepsilon}_n \}$, & $\widehat S_\varepsilon$,  vertices in $conv(S_\varepsilon)$,  $| \widehat S_\varepsilon|=K_\varepsilon$, &  $\sigma \leq \sigma_* \equiv \Sigma_*/R$ is known, $\varepsilon \leq \sigma/4$ & $O \big (nK_\varepsilon(m+\sigma^{-2}) \big )$
			\\
			a perturbation of $S$
			
			& $\overline S_\varepsilon=\{\overline v^{\varepsilon}_1, \dots, \overline v^{\varepsilon}_K \} \subset \widehat S_\varepsilon$ & $\gamma \leq \gamma_*$ is known, $\varepsilon \leq \gamma \rho_*/4R$ & $O \big (nK_\varepsilon(m+ R^2/(\gamma \rho_*)^{2}) \big )$
			
			\\
			
			& & Only $K$ is known, $\varepsilon \leq \sigma_*/4$ & $O \big (nK_\varepsilon(m+ \sigma_*^{-2}) \big ) \times \log (\sigma_*^{-1})$

			\\
			$\Vert v^{\varepsilon}_i - v_i \Vert \leq \varepsilon R$ &
			Given $t  \in (0,1)$, $\overline S^t_\varepsilon \subset \overline S_\varepsilon$, $|\overline S^t_\varepsilon|=K^{(t)}_\varepsilon$
			&General Case&  $O \big (nK^{(t)}_\varepsilon(m+t^{-2}) \big )$
			
			\\
			& $d(p, conv(\overline S_\varepsilon^t) \leq (t+\varepsilon) R, \forall  p \in conv(S)$  & &
			\\
			\hline
			& & &
			\\
			
			J-L Projection of $S$& $\overline U$, vertices of $conv(U)$, $|\overline U|=K_{\varepsilon'}$  &  $\gamma \leq \gamma_*$ is known& $O \big (nK_{\varepsilon'}(m'+ (\gamma (1- \varepsilon'))^{-2}) \big )$
			
			\\
			$U=L(S)=\{u_1, \dots, u_n \}$ & $\overline U=\{\overline u_1, \dots, \overline u_{K_{\varepsilon'}} \}$, $\overline U \subset L(\overline S)$ &  $m' =\varepsilon'^2/ c \log n < m$, $c$ a constant &
			\\
			$u_i=L(v_i)$, $L : \mathbb{R}^{m} \rightarrow  \mathbb{R}^{m'}$ & Given $t \in (0,1)$, $\overline U^t \subset L(\overline S)$, $|\overline U^t|=K_{\varepsilon'}^t$
			&  General Case & $O \big (nK_{\varepsilon'}^t(m'+ t^{-2}) \big )$
			
			\\
			$R'$ diameter of $U$ & $d(q, conv(\overline U^t) \leq t R', \forall q \in conv(U)$& &
			\\
			\hline
			& & &
			\\
			
			J-L Projection of $S_\varepsilon$& $\widehat U_\varepsilon$, vertices in $conv(U_\varepsilon)$, $|\widehat U_\varepsilon|=K_{\varepsilon \varepsilon'}$, & $\sigma \leq \sigma_*$ is known, $\varepsilon \leq \sigma (1- \varepsilon')/4$ & $O \big (nK_{\varepsilon \varepsilon'}(m'+ (\sigma (1- \varepsilon'))^{-2}) \big )$
			
			\\
			
			$U_\varepsilon= L(S_\varepsilon)= \{u_1^{\varepsilon}, \dots, u_n^{\varepsilon} \}$ &
			$\overline U_\varepsilon=\{\overline u^{\varepsilon}_1, \dots, \overline u^{\varepsilon}_{K_*} \}$,
			$\overline U_\varepsilon \subset \widehat U_\varepsilon$
			& $\gamma \leq \gamma_*$ is known, $\varepsilon \leq \gamma \rho_* (1- \varepsilon')/4R$ & $O \big (nK_{\varepsilon \varepsilon'}(m'+ R^2/(\gamma \rho_*(1- \varepsilon'))^{2}) \big )$
			
			\\
			
			$u^{\varepsilon}_i= L(v^{\varepsilon}_i)$
			&
			Given $t \in (0,1)$, $\overline U^t_\varepsilon \subset \overline U_\varepsilon$, $|\overline U^t_\varepsilon|=K^{(t)}_{\varepsilon \varepsilon'}$ &
			General Case & $O \big (nK^{(t)}_{\varepsilon \varepsilon} (m'+ t^{-2}) \big )$
			
			\\
			& $d(q, conv(\overline U_\varepsilon^t) \leq (t+ \varepsilon) R, \forall q \in conv(U)$& &
			\\
			\hline
		\end{tabular}
	}
	\caption{{\small $\Gamma_*= \min \{d(\overline v_i, conv(\overline S \setminus \{\overline v_i\}))\}$, $\Sigma_*= \min \{d(\overline v_i, conv(S \setminus \{\overline v_i\}))\}$, $\rho_* = \min \{d(v_i, v_j), i \not = j\}$}.}
	\label{Table1}
\end{table}

We also consider  the application of AVTA in the recovery of vertices through the  projection of $S$ or $S_\varepsilon$ under  a Johnson-Lindenstrauss randomized linear projection $L : \mathbb{R}^{m} \rightarrow  \mathbb{R}^{m'}$.
By relating the robustness parameters of $conv(U)$ and $conv(U_\varepsilon)$, where $U=L(S)$  and $U_\varepsilon=L(S_\varepsilon)$, to those of $conv(S)$ and $conv(S_\varepsilon)$, we derive analogous complexity bounds for probabilistic computation of the vertex set of $conv(U)$ or those of $conv(U_\varepsilon)$, or an approximation to these subsets for a given $t \in (0,1)$.  Table ~\ref{Table1} summarizes the  complexities of computing desired sets under various cases.

The organization of the the article is as follows. In Section \ref{TA}, we review the Triangle Algorithm  for solving the convex hull membership problem.  In Section \ref{EffTA}, we describe an efficient implementation of the Triangle Algorithm.  This will be used throughout the the article.  In Section \ref{AVTASEC}, we describe {\it All Vertex Triangle Algorithm} (AVTA), a modification of the Triangle Algorithm, for computing all vertices of the convex hull of a given finite set of points, $S$.  We discuss several applications of this, in particular in solving the convex hull membership problem itself.  Other applications will be described in subsequent sections.  In Section \ref{AVTAPER}, we consider the performance of AVTA under perturbation of data.  In Section \ref{sec6}, we consider AVTA with Johnson-Lindenstrauss projections. Furthermore, we consider the performance of AVTA under perturbation of data with Johnson-Lindenstrauss projections. 

\section{Review of The Triangle Algorithm} \label{TA}

The {\it Triangle Algorithm} described in \citet{kalantari2015characterization} is a simple iterative algorithm for solving the {\it convex hull membership problem}, a fundamental problem in linear programming and computational geometry. Formally, the convex hull membership problem is as follows:  Given a set of point $S=\{v_1, \dots, v_n\} \subset \mathbb{R}^m$ and a distinguished point $p \in \mathbb{R} ^m$, test if $p \in conv(S)$. If $p \not \in conv(S)$,  find a hyperplane that separates $p$ from $conv(S)$.   If $p \in conv(S)$, the Triangle Algorithm solves the problem to within prescribed precision by generating a sequence of points inside of $conv(S)$ that get sufficiently close to $p$.

Given two point $u,v \in \mathbb{R} ^m$ we interchangeably use
$d(u,v)= \Vert u- v \Vert$.  Given a point in  $conv(S)$, the Triangle Algorithm  searches for a {\it pivot} to get closer to $p$:

\begin{definition} \label{pivot} Given $p' \in conv(S)$, called {\it iterate},  we call $v \in S$  a $p$-{\it pivot} (or simply {\it pivot}) if
\begin{equation} \label{eq1}
d(p', v) \geq d(p, v).
\end{equation}
Equivalently, $v$ is a pivot if and only if
\begin{equation} \label{eq2}
v^Tp-v^Tp' \geq  \frac{1}{2} (\Vert p \Vert^2 -  \Vert p' \Vert^2).
\end{equation}
\end{definition}

\begin{definition} \label{witness}
A point $p' \in conv(S)$ is a $p$-{\it witness} (or simply {\it witness}) if the orthogonal bisecting hyperplane to the line segment $pp'$ separates $p$ from $conv(S)$.
\end{definition}
Equivalently,  $p' \in conv(S)$ is a  $p$-{\it witness} if and only if
\begin{equation} \label{eq3}
d(p',v_i) < d(p,v_i), \quad \forall i=1, \dots, n.
\end{equation}

\begin{definition}
Given $\varepsilon \in (0,1)$, $p' \in conv(S)$  is an $\varepsilon$-{\it approximate solution} if for some $v \in S$,
\begin{equation} \label{eq4}
d(p',p) \leq \varepsilon R,
\end{equation}
where $R$ is the diameter of $S$:
\begin{equation}
R= \max \{d(v_i, v_j): v_i,  v_j \in S\}.
\end{equation}

\end{definition}

Given a point $p' \in conv(S)$ that is neither an $\varepsilon$-approximate solution nor a witness, the Triangle Algorithm finds a $p$-pivot $v \in S$. Then on the line segment $p'v$ it compute the closest point to $p$, denoted by $Nearest(p; p'v)$. It then replaces $p'$ with $Nearest(p; p'v)$ and repeats the process.  It is easy to show,

\begin{prop} \label{prop1}
If an iterate $p' \in conv(S)$ satisfies $d(p', p) \leq \min \{d(p, v_i): i=1, \dots, n\}$, and $v_j$ is a $p$-pivot, then the new iterate
\begin{equation} \label{eq5}
p'' =Nearest(p;p'v_j)=
(1-\alpha)p' + \alpha v_j,
\end{equation}
where the {\it step-size} $\alpha$ is
\begin{equation} \label{eq6}
\alpha = \frac{(p-p')^T(v_j-p')}{d^2(v_j,p')} = \frac{p^Tv_j -p'^Tv_j - p^T p' + \Vert p' \Vert^2}{\Vert v_j \Vert^2 - 2p'^T v_j + \Vert p' \Vert^2}.
\end{equation}
In particular if
\begin{equation} \label{eq7}
p'=\sum_{i=1}^n \alpha_i v_i, \quad \sum_{i=1}^n \alpha_i=1, \quad \alpha_i \geq 0, \quad \forall i,
\end{equation}
then
\begin{equation} \label{eq8}
p''=\sum_{i=1}^n \alpha'_i v_i
\end{equation}
where
\begin{equation} \label{eq9}
\alpha'_j=(1-\alpha)\alpha_j+\alpha,  \quad \alpha'_i= (1-\alpha)\alpha_i,  \quad \forall i \not =j.  \qed
\end{equation}
\end{prop}

The following duality ensures the correctness of the iterative step of the Triangle Algorithm.

\begin{thm}  \label{thm1} {\bf (Distance Duality)}  Precisely one of the two conditions is satisfied:\

(i) For each  $p' \in conv(S)$ there exists $v_j \in S$ that is $p$-{\it pivot}, i.e. $d(p',v_j) \geq d(p,v_j)$.

(ii) There exists $p' \in conv(S)$ that is $p$-{\it witness}, i.e. $d(p',v_i) < d(p, v_i)$ for all $v_i \in S$.  $\qed$
\end{thm}

The following  relates the gap in two consecutive iterates of the Triangle Algorithm:

\begin{thm}  \label{thm2}  Let $p, p', v$  be distinct points in $\mathbb{R} ^m$. Suppose $d(p,p') \leq d(p,v) \leq d(p',v)$.  Let $p''=Nearest (p, p'v)$. Let $\delta=d(p',p)$, $\delta'=d(p'',p)$, and $r=d(p,v)$. Then,

\begin{equation} \label{gap}
\delta' \leq \delta \sqrt{1- \frac{\delta^2}{4 r^2}}. \qed
\end{equation}
\end{thm}

\begin{figure}[htpb]
	\centering
	\begin{tikzpicture}[scale=0.3]

	
		\draw (0.0,0.0) -- (7.0,0.0) -- (-2.0,-4.0) -- cycle;
      \draw (0,0) -- (7,0) node[pos=0.5, above] {$r$};
      \draw (-2,-4) -- (0,0) node[pos=0.5, above] {$\delta$};
       \draw (0,0) -- (1.15,-2.6) node[pos=0.5, right] {$\delta'$};
       \draw (1.15,-2.6) node[below] {$p''$};
       \filldraw (1.15,-2.6) circle (2pt);
		\draw (0,0) node[left] {$p$};
		\draw (7,0) node[right] {$v$};
		\draw (-2,-4) node[below] {$p'$};
           \filldraw (0,0) circle (2pt);
\filldraw (7,0) circle (2pt);
\filldraw (-2,-4) circle (2pt);
		
	\end{tikzpicture}
\begin{center}
\caption{Reduction of gaps $\delta=\Vert p'- p \Vert$ by using a $p$-pivot $v$, $\delta'=\Vert p''-p \Vert$.} \label{Fig5}
\end{center}
\end{figure}
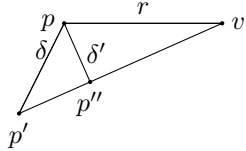

The following gives the aggregate complexity bound.

\begin{thm}   \label{thm3}  The Triangle Algorithm correctly solves the convex hull membership problem as follows:

(i) Suppose $p \in conv(S)$. Given $\varepsilon >0$, the number of iterations $K_\varepsilon$ to compute $p_\varepsilon \in conv(S)$ so that  $d(p,p_\varepsilon) \leq \varepsilon d(p, v_i)$, for some $v_i \in S$ is
\begin{equation} \label{iter1}
K_\varepsilon  \leq \frac{48}{\varepsilon^2}=  O  \bigg (\frac{1}{\varepsilon^2} \bigg ).
\end{equation}
(ii) Suppose $p \not \in conv(S)$. Let $R=max\{d(v_i,v_j):  v_i, v_j \in S\}$,
$\Delta= \min \{d(x, p) : x \in conv(S) \}$. The number of iterations $K_\Delta$ to compute
$p_\Delta \in conv(S)$ so that $d(p_\Delta,v_i) < d(p, v_i)$ for all $v_i \in S$, satisfies
\begin{equation} \label{iter2}
K_\Delta \leq \frac{48 R^2}{\Delta^2}= O  \bigg (\frac{R^2}{\Delta^2} \bigg ). \qed
\end{equation}
\end{thm}

The straightforward implementation of each iteration of the Triangle Algorithm is easily seen to  take $O(mn)$ arithmetic operations. The algorithm can be described as follows:

\begin{center}
\begin{tikzpicture}
\node [mybox] (box){%
    \begin{minipage}{0.9\textwidth}
{\bf  Triangle Algorithm ($S$, $p$, $\varepsilon \in (0,1)$)}\
\vspace{.2cm}

\begin{itemize}
\item
{\bf Step 0.} Set $p'= {\rm argmin} \{d(v_i,p)\}$.

\item
{\bf Step 1.} If $d(p',p) \leq \varepsilon R$, then output $p'$ as an $\varepsilon$-approximate solution, stop.

\item
{\bf Step 2.}  If a $p$-pivot $v \in S$  exists, set $p' \leftarrow Nearest(p; p'v)$. Goto Step 1.

\item
{\bf Step 3.} Output $p'$ as $p$-witness. Stop.

\end{itemize}
\end{minipage}};
\end{tikzpicture}
\end{center}

\begin{remark}
In each iteration of the Triangle Algorithm it suffices to have a representation of the iterate $p'$ in terms of $v_i$'s, i.e. $p'=\sum_{i=1}^n \alpha_i v_i$,  where $\sum_{i=1}^n \alpha_i=1$, $\alpha_i \geq 0$ for all $i=1, \dots, n$.
It is not necessary to know the coordinates of $p'$.  Rather it is enough to have an array of size $n$ to store the vector of $\alpha_i$'s.  Then assuming that we have stored $p^Tv_i$, $i=1, \dots, n$, we can compute
the step size $\alpha$ (see (\ref{eq6})) and $p''$ (the new iterate) in $O(n)$ time.
\end{remark}

An alternate complexity bound can be stated for the Triangle Algorithm,  especially when $p \in conv(S)$ is well situated.

\begin{definition} \label{strictpivot} Given $p' \in conv(S)$,  $v \in S$ is a {\it strict} $p$-pivot (or simply {\it strict} pivot) if $\angle p'pv \geq \pi/2$.
\end{definition}

\begin{thm}  \label{thm4} {\bf (Strict Distance Duality) } Assume $p \not \in S$.  Then $p \in conv(S)$ if and only if for each  $p' \in conv(S)$ there exists strict $p$-pivot, $v \in S$.  $\qed$
\end{thm}

The following theorem shows that under the assumption that $p$ is an interior point of $conv(S)$ we can give an alternate complexity for the Triangle Algorithm whose number of iterations are logarithmic in $1/\varepsilon$.

\begin{thm}  \label{thm5} Suppose $p$ is contained in a ball of radius $\rho$,  $B_\rho (p)= \{x:  d(x, p) < \rho\}$, contained in  $conv^\circ(S)$, the relative interior of $conv(S)$. Suppose the Triangle Algorithm uses a strict pivot in each iteration.  Given $\varepsilon \in (0,1)$, the number of iterations to compute $p_\varepsilon \in conv(S)$ such that $d(p,p_\varepsilon) \leq \varepsilon R$, $R=max\{d(v_i,v_j), v_i, v_j \in S\}$ is
\begin{equation}
O\bigg ( \bigg (\frac {R}{\rho} \bigg )^2\log \frac{1}{\varepsilon} \bigg ). \qed
\end{equation}
\end{thm}

\section{Efficient Implementation of Triangle Algorithm} \label{EffTA}
The worst-case complexity of each iteration in the Triangle Algorithm is $O(mn)$.  Assuming that all the inner products $v_i^T  v_j$ are computed the iteration complexity of Triangle Algorithm can be shown to reduce to $O(n)$.  The cost of pre-computing the inner products is $O(mn^2)$.  The complexity can be made more efficient. To do so it suffices to compute the inner products $ v_i^T v_j$  progressively rather than pre-computing them all.  Ignoring this complexity,  the iteration complexity of the Triangle Algorithm reduces to $O(N)$ where $N \leq n $ is the number of points of $S$ considered in the Triangle Algorithm. The following shows how to achieve this.

\begin{prop}  \label{prop2} Let $\widehat S= \{ \widehat v_1, \dots, \widehat v_N\}$ be a subset  of  $\mathbb{R}^m$.
 Consider testing if a given $p \in \mathbb{R}^m$ lies in $conv( \widehat S)$.  Suppose we have computed  $\Vert p \Vert ^2$,  as well as $p^T \widehat v_i$, $i=1, \dots, N$.  Suppose we have available $p' = \sum_{i=1}^N \alpha_i \widehat v_i \in conv( \widehat S)$ satisfying $d(p', p) \leq \min \{d(p, \widehat v_i): i=1, \dots, N\}$. Suppose  $\Vert p' \Vert ^2$ is also computed. Also, suppose  $p'^T \widehat v_i$ is computed for each $i=1, \dots, N$.  Then excluding the cost of computing the  entries of the $N \times N$ matrix $\widehat M=(\widehat v_i^T \widehat v_j)$, each iteration of the Triangle Algorithm can be implemented in $O(N)$ operations. More precisely,

(i) Computation of a $p$-pivot $\widehat v_j$ at $p'$, if one exists,  takes $O(N)$ operations.

(ii) Given a pivot $\widehat v_j$, the computation of step size $\alpha$ takes $O(1)$ operations.

(iii) Computation of $Nearest(p; p'v)=p''= (1- \alpha) p' + \alpha \widehat v_j = \sum_{i=1}^N \alpha'_i \widehat v_i$
takes $O(N)$ operations.

(iv) Computation of $\Vert p'' \Vert^2$ takes $O(1)$ operations.

(v) Computation of
$p''^T \widehat v_i$, $i=1, \dots, N$  takes $O(N)$ operations.
\end{prop}

\begin{proof} The Triangle Algorithm needs to use the entries of the $N \times N$ matrix $\widehat M=(\widehat v_i^T \widehat v_j)$. However, not all entries may be needed, nor do all entries of $\widehat M$ need to be computed in advance. Putting aside  the complexity of computing $\widehat M$, in  the following we justify the claimed complexities.

(i): From (\ref{eq2}) and the given assumptions, to check if a particular $\widehat  v_i$ is  a pivot takes $O(1)$ operations. Thus to check if there exists a pivot takes $O(N)$ time.

(ii): From (\ref{eq6}) and the assumptions, to compute $\alpha$ takes $O(1)$ operations.

(iii): From equations (\ref{eq8}) and (\ref{eq9}) the computation of $p''$ and its representation
as $p''=\sum_{i=1}^N \alpha'_i \widehat v_i$ takes $O(N)$ operations.

(iv): Since $p''=(1- \alpha)p' + \alpha \widehat v_j$, we have
\begin{equation} \Vert p'' \Vert^2= p''^Tp''= (1-\alpha)^2 \Vert p' \Vert^2 + 2\alpha (1- \alpha) p'^T \widehat v_j + \alpha^2 \Vert \widehat v_j \Vert ^2.\end{equation}
It follows that  computing $\Vert p'' \Vert ^2$ takes $O(1)$ operations.

(v): Using that  $p''=(1- \alpha)p' + \alpha \widehat v_j$, the computation of $p''^T \widehat v_i$ takes $O(1)$ computations. Hence to compute all inner products  $p''^T \widehat v_i$, $i=1, \dots, N$   takes $O(N)$ computations.
\end{proof}

The following theorem combines Theorem \ref{thm3} and Proposition \ref{prop2} giving an improved complexity for the Triangle Algorithm.

\begin{thm}  \label{thm6}   Let $\widehat S= \{ \widehat v_1, \dots, \widehat v_N\}$ be a subset  of $S=\{v_1, \dots, v_n\}$. Given $p \in \mathbb{R}^m$, consider testing if $p \in conv( \widehat S)$.  Suppose $\Vert p \Vert ^2$  as well as $p^T \widehat v_i$, $i=1, \dots, N$ are computed.
Given $\varepsilon \in (0, 1)$, assume the Triangle Algorithm starts with $p' = {\rm argmin} \{ d(\widehat v_i, p) : i=1, \dots, N\}$. Then the complexity of testing if there exists an $\varepsilon$-approximate solution is
\begin{equation} \label{eqAA}
O \bigg (m N^2+ \frac{N}{\varepsilon^2} \bigg).
\end{equation}

In particular, suppose in testing if $p \in conv(S)$, $S=\{v_1, \dots, v_n\}$, the Triangle Algorithm computes an $\varepsilon$-approximate solution $p_\varepsilon$ by examining only the elements of a subset $\widehat S= \{ \widehat v_1, \dots, \widehat v_N\}$ of $S$.  Then the  number of operations to determine if  there exists an $\varepsilon$-approximate solution  $p_\varepsilon \in conv(S)$, is  as stated in (\ref{eqAA}). $\qed$
\end{thm}

\section{All Vertex Triangle Algorithm (AVTA)} \label{AVTASEC}
Given $S= \{v_i \in \mathbb{R} ^m: i=1, \dots, n\}$, let $R$ be its diameter, i.e. $R= \max \{d(v_i, v_j), v_i,v_j \in S\}$.  Denote the set of vertices of $conv(S)$ by
\begin{equation} \overline S =\{\overline v_1, \dots, \overline v_K\}. \end{equation}
A straightforward but naive way to compute $\overline S$  is to test for each $v_i$ if it lies in $conv(S \setminus \{v_i\})$, to within an $\varepsilon$ precision.  Thus the overall this would take $n$ times the complexity of Triangle Algorithm.   This is inefficient. In what follows we describe a modification of the Triangle Algorithm with more efficient complexity than the straightforward algorithm. First we give a definition.

\begin{definition}  \label{defAA} We say $conv(S)$ is $\Gamma_*$-{\it robust} if
\begin{equation}
\Gamma_*= \min \{ d(\overline v_i,  conv(\overline S \setminus \{\overline v_i\})):  i=1, \dots, K \}.
\end{equation}
\end{definition}

As an example, given a triangle with vertices $v_1, v_2, v_3$,  $\Gamma_*$ is the minimum of the distances from each vertex to the line segment determined by the other vertices.  Thus if other points are placed inside the triangle $\Gamma_*$ will not be affected.

The following is immediate from Definition \ref{defAA}.

\begin{prop}  \label{prop3}   Let  $\widehat S= \{\widehat v_1, \dots, \widehat v_N \}$ be a subset of $\overline S$.  Suppose $conv(S)$ is $\Gamma_*$-robust.  Given $v \in S \setminus \widehat S$,  if for some $\gamma \leq \gamma_* \equiv \Gamma_*/R$ we have
$d (v, conv(\widehat S)) < \gamma R$, then
$v \not \in \overline S$. $\qed$
\end{prop}

\begin{thm}  \label{thm7} Let $\widehat S= \{ \widehat v_1, \dots, \widehat v_N\}$ be a subset of $\overline S$.  Given $\gamma \in (0,1)$,  consider testing if a given  $v \in S \setminus \widehat S$  satisfies $d(v, conv(\widehat S)) \leq \gamma R /2$. Suppose  we are given $p' \in   conv( \widehat S)$ for which $\Vert p' \Vert ^2$  as well as $p'^T \widehat v_i$, $i=1, \dots, N$ are computed. Then the number of operations to check if $d(v, conv(\widehat S)) \leq \gamma R/2$ satisfies
\begin{equation}O \bigg(m K^2+ \frac{K}{\gamma^2} \bigg ).\end{equation}
\end{thm}

\begin{proof} Proof is immediate from  Theorem \ref{thm6} and the fact that $N \leq K$.
\end{proof}

We now describe a modification of the Triangle Algorithm for computing all vertices of $conv(S)$. We call this {\it All Vertex Triangle Algorithm} or simply AVTA.  Assume $conv(S)$ is $\Gamma_*$-robust, where $\Gamma_*$ may or may not be available.  However, assume we have a constant $\gamma \in (0,1)$ known to satisfy $\gamma \leq \gamma_*=\Gamma_*/R$.  AVTA  works as follows.  Given a working subset $\widehat S$ of $\overline S$,  initially of cardinality $N=1$ (see Proposition \ref{prop4}),  a single vertex of $S$, it randomly selects $v \in S \setminus \widehat S$. It then tests via the Triangle Algorithm if $d(v, conv(\widehat S)) \leq \gamma R/2$. If so, it discards $v$ since by definition of $\gamma$ it cannot belong to $\overline S$ (see Proposition \ref{prop3}). Otherwise,  it computes a $v$-witness $p' \in conv(\widehat S)$. It then sets $c'=v-p'$ and maximizes $c'^Tx$ where $x$ ranges in $conv(S \setminus \widehat S)$. The maximum  value coincides with the maximum of  $c'^Tv_i$ where $v_i$ ranges  in $S \setminus \widehat S$.  If the set of optimal solution $S \setminus \widehat S$ is denoted by $S'$,  then $conv(S')$ is a face of $conv(S)$.  A vertex $v'$ of $conv(S')$ is a point in $S'$ and  is  necessarily a vertex of $conv(S)$.  Such a vertex can be computed efficiently.  Having computed a new vertex $v'$ of $conv(S)$, AVTA replaces  $\widehat S$ with  $\widehat S \cup \{v'\}$ and the process is repeated. However,  if $v$ coincides with $v'$  AVTA selects a new point in $S \setminus \widehat S$.  Otherwise,  AVTA continues to test if the same $v$ (for which a witness was found) is within a distance of $\gamma R/2$ of the convex hull of the augmented set $\widehat S$. Also, as an iterate AVTA uses the same witness $p'$.
In doing so  each selected $v \in S$  is either determined to be a vertex itself, or it will continue to be tested if it is lies to within a distance of $\gamma R/2$ of the growing set $\widehat S$. If within $\gamma R/2$ distance, it will be discarded before AVTA tests another point. We will describe AVTA more precisely.  However, we first prove the necessary results.

\begin{lemma}  \label{lem1}  Let $\widehat S= \{ \widehat v_1, \dots, \widehat v_N\}$ be a subset of $\overline S$. For a given $v \in S \setminus  \widehat S$ suppose $p' \in conv(\widehat S)$ is a $v$-witness.  Let $c'=v-p'$.  Then
\begin{equation}\max \{c'^T x:  x \in conv(S \setminus  \widehat S) \}= \max \{c'^T v_i:  v_i \in S \setminus  \widehat S \}.\end{equation}
\end{lemma}
\begin{proof}  Each $x  \in conv(S \setminus  \widehat S)$ can be written as a convex combination
\begin{equation}x = \sum_{i: v_i \in  S \setminus  \widehat S} \alpha_i v_i, \quad \sum_{i: v_i \in  S \setminus  \widehat S} \alpha_i =1, \quad \alpha_i \geq 0.\end{equation}
Then
\begin{equation}c'^T x= \sum_{i: v_i \in  S \setminus  \widehat S} \alpha_i c'^T v_i, \quad \sum_{i: v_i \in  S \setminus  \widehat S} \alpha_i =1, \quad \alpha_i \geq 0.\end{equation}
It follows that the maximum of $c'^Tx$ over $S \setminus  \widehat S$ can be computed trivially.
\end{proof}

\begin{cor} Let $c'=v-p'$ be as in Lemma \ref{lem1}, $p' \in   conv( \widehat S)$ for which $\Vert p' \Vert ^2$  as well as $p'^T \widehat v_i$, $i=1, \dots, N$ are computed. Then,
$\max \{c'^T x:  x \in conv(S \setminus  \widehat S) \}$ can be computed in $O(nK)$ operations.
\end{cor}

\begin{proof} Since $N \leq K$, for each $i$, $c'^Tv_i$ can be computed in $O(K)$ operations.
\end{proof}

\begin{thm}  \label{thm8} Let $S'$ be the set of optimal solutions of $\max \{c'^T x:  x \in S \setminus  \widehat S \}$. Let $v' \in S'$ be a vertex of $conv(S')$.  Then $v'$ is a vertex of $conv(S)$, i.e. $v \in \overline S =\{ \overline v_1, \dots, \overline v_K\}$.
\end{thm}

\begin{proof} We can write $v'$ as a  convex combination of $\overline v_i$, $i=1, \dots, K$:
\begin{equation}v'= \sum_{i=1}^K \alpha_i \overline v_i, \quad \sum_{i=1}^K \alpha_i =1, \quad \alpha_i \geq 0, \quad \forall i.\end{equation}
The above can be rewritten as
\begin{equation}v'= \sum_{i:  \overline v_i \in S'} \alpha_i \overline v_i  + \sum_{i:  \overline v_i \not \in S'} \alpha_i \overline v_i  .\end{equation}
Since for $v_i \in S'$, $c'^T v' = c'^Tv_i$ and for $v_i \not \in S'$, $c'^Tv' > c'^T v_i$, it follows that $v'$ is a convex combination of $ \overline v_i$ for which $c'^T v'= c'^T \overline v_i$.  But since $v'$ is a vertex of $conv(S')$ it follows that $v' \in \overline S$.
\end{proof}

The following shows computing a  single vertex of $conv(S)$ is trivial.

\begin{prop} \label{prop4}  Given any $v$ in $S$, let $Farthest(v, S)$ return a point in $S$  that is farthest from $v$.  Then $Farthest(v,S)$ is a vertex of $conv(S)$, hence a member of $\overline S$.
\end{prop}

\begin{proof}  If $Farthest(v,S)$ is not a vertex of $conv(S)$ it can be written as a convex combination of two other points $v_1, v_2 \in conv(S)$. But then this gives a contradiction by considering the triangle with vertices $v_1,v_2, Farthest(v, S)$.
\end{proof}

When $n \geq 2$, $Farthest(v, S)$ for all $v \in S$ results in a set with at least two points but it may also contain exactly two points.  It can be computed in $O(n^2)$ time.  Next we describe AVTA for computing all vertices of $conv(S)$.

\begin{center}
\begin{tikzpicture}
\node [mybox] (box){%
    \begin{minipage}{0.9\textwidth}
{\bf  AVTA ($S$, $\gamma \in (0,1)$)}\
\vspace{.2cm}
\begin{itemize}
\item

{\bf Step 0.} Set $\widehat S = \{Farthest(v, S)\}$ for some $v \in S$.

\item {\bf Step 1.} Randomly select $v \in S \setminus \widehat S$.

\item
{\bf Step 2.} Call {\bf Triangle Algorithm} $(\widehat S, v,  \gamma/2)$.

\item
{\bf Step 3.}
If the output $p'$ of Step 2 is a $v$-witness then Goto Step 4. Otherwise, $p'$ is a $\gamma/2$-approximate solution to $v$. Set $S \leftarrow S \setminus \{v\}$. If $S= \emptyset$, stop. Otherwise, Goto Step 1.


\item
{\bf Step 4.}
Let $c'=v-p'$.

Compute $S'$, the set of optimal solutions of $\max \{c'^T x:  x \in S \setminus  \widehat S \}$. Randomly select $v' \in S'$. $ v' \leftarrow Farthest(v', S')$, $\widehat S \leftarrow \widehat S \cup \{v'\}$.

\item
{\bf Step 5.} If $v = v'$, Goto Step 1. Otherwise, Goto Step 2.

\end{itemize}
\end{minipage}};
\end{tikzpicture}
\end{center}

\begin{remark} Here we make remarks about the steps of AVTA.   In Step 0 AVTA selects the first vertex. In Step 1 it randomly select a $v$ in $S \setminus \widehat S$.  In Step 2 AVTA checks if the point $v$ selected in Step 1 is sufficiently close to the convex hull of the current set of vertices, $\widehat S$. If so, in Step 3 $v$  is discarded from further considerations.  Otherwise, a $v$-witness  $p'$ is at hand.  Step 4 then uses this witness to compute a direction, $c'=v-p'$, where the maximization of $c'^Tx$ gives a subset $S'$ of $S$ consisting of the optimal solutions.  Then a vertex $conv(S')$ will necessarily be a vertex of $conv(S)$.  A vertex of $conv(S')$ is selected by choosing an arbitrary $v' \in S'$ and computing its farthest point in $S'$.   It maybe the case that the vertex $v'$ found in Step 4 coincides with $v$.  Step 5 checks if $v'=v$ in which case it select a new $v$ in the updated $S \setminus \widehat S$ in Step 1 for consideration.   Otherwise, when this new vertex $v'$ is not $v$ itself, in Step 5 in AVTA $v$ is sent back to Step 2 to be reexamined  if $v$ is within $\gamma R/2$ distance of the convex hull of augmented $\widehat S$.
\end{remark}

\begin{example}
We consider an example of AVTA, see Figure \ref{Fig1}. In this example $S=\{v_1, \dots, v_{11}\}$.  Note that the set of vertices is $\overline S=\{v_4,v_{10}, v_6, v_1, v_9, v_2, v_5, v_8\}$.  Suppose the current working subset of vertices  $\overline S$ consists of  $\widehat S=\{v_1, v_9, v_2, v_5\}$ and $v=v_3$ is randomly selected to be tested if it lies in $conv(\widehat S)$.  A witness $p' \in conv(\widehat S)$ is computed and with $c'=p'-v$ maximum of $c'^Tx$ over $conv(\overline S \setminus \widehat S)$ is attained at $S'=\{v_4, v_7, v_{10}\}$. Subsequently one of the two points $v_4$ or $v_{10}$ will become the next vertex to be placed in $\widehat S$.

\begin{figure}[htpb]
	\centering
	
	\begin{tikzpicture}[scale=0.7]
			
\draw (0.0,0.0) -- (8,2.0) --(7,0)--(5,-2)-- cycle;
		\draw (0,0) node[below] {$v_1$};
		\draw (7,0) node[right] {$v_2$};
		\draw (4,5) node[above] {$v=v_3$};
		\draw (5,-2) node[below] {$v_{9}$};
		\draw (8,2) node[right] {$v_5$};
      \filldraw (5,8.) circle (2pt) node[above] {$v_7$};
       \filldraw (7.6,8.) circle (2pt) node[above] {$v_4$};
            \filldraw (1.5,8.) circle (2pt) node[above] {$v_{10}$};
\filldraw (5,-2) circle (2pt);
\filldraw (0,0) circle (2pt);
\filldraw (8,2) circle (2pt);
\filldraw (7,0) circle (2pt);
\filldraw (4,5) circle (2pt);
\filldraw (.3,4.5) circle (2pt) node[above] {$v_6$};
\filldraw (4,0) circle (2pt) node[below] {$p'$};
\filldraw (2,3) circle (2pt) node[left] {$v_{11}$};
\draw (4,0) -- (4,5);
\draw (-1,2.5) -- (10,2.5);
\draw (-1,8) -- (10,8);
\begin{scope}[red]
\end{scope}
\filldraw (8,6)  circle (2pt) node[right] {$v_8$};
	
	\end{tikzpicture}
	
	\caption{An example where $ \widehat S =\{v_1, v_9, v_2, v_5\}$, then  $v=v_3$ is randomly selected from $S \setminus \widehat S$ and is tested if it lies in $conv(\widehat S)$. A witness $p'$ is found.  Then using $c'-v-p'$ the set $S'=\{v_4, v_7, v_{10}\}$ is computed and one of vertices of $conv(S')$, i.e. $v_4$ or $v_{10}$ is selected for inclusion in $\widehat S$.}
	\label{Fig1}
\end{figure}
\end{example}

The following theorem is one of the main results:

\begin{thm}  \label{thm9} Let $S= \{v_1, \dots, v_n\} \subset \mathbb{R} ^m$.  Let $R$ be the diameter of $S$. Let $ \overline S =\{\overline v_1, \dots, \overline v_K\}$ be the set of vertices of $conv(S)$.  Suppose that $conv(S)$ is $\Gamma_*$-robust. Let $\gamma_*=\Gamma_*/R$.

(1) If a number $\gamma$ is known satisfying   $0 < \gamma \leq \gamma_*$, the number of arithmetic operations to compute $\overline S$ is
\begin{equation}O \bigg (nK(m+ \frac{1}{\gamma^{2}}) \bigg ).\end{equation}

(2) If only $K$ is known, the complexity of computing $\overline S$ is
\begin{equation}O\bigg ( \bigg (nK(m+ \frac{1}{\gamma_*^2} ) \bigg ) \log \frac{1}{\gamma_*} \bigg ).\end{equation}

(3)  More generally, given any prescribed $t \in (0,1)$ in
\begin{equation}O \bigg (nK^{(t)}(m+ \frac{1}{t^2})\bigg )\end{equation}
operations AVTA computes a subset $\overline S^t$ of $\overline S$  of size $K^{(t)}$ so that the distance from each $p$ in $conv(S)$  to $conv(\widehat S)$ is at most $t R$.
\end{thm}

\begin{proof}  (1):  Initially in AVTA the subset $\widehat S$ consists of a single element of $\overline S$. It continues to grow until it reaches $\overline S$. By Theorem \ref{thm6} for each $v \in S \setminus \widehat S$ the cost of Step 2 in AVTA is $O(mK^2+ {K}/{\gamma^{2}})$.    The needed inner products in Step 2 are $\widehat v_i^T \widehat v_j$.  However, these inner products need to be computed only once  and  since there are at most $K$ of $\widehat v_i$'s,  these inner products can be computed at the cost of $O(mK^2)$ operations. We can store the values of the inner products in an array. Then we use them again as they arise in subsequent iterations. This kind of storing can be done for other inner products that may need to be computed in the course of the algorithm.
When a selected $v$ is within the distance of $\gamma R/2$ to $conv(\widehat S)$, Step 3 eliminates it from further considerations.  If $v$ is not eliminated, it either gives rise to a new vertex $v' \in \overline S$, or $v$ is a vertex itself.  In either case, in order to identify a new vertex of $\overline S$, after a witness has become available, it requires the minimization of $c'^Tv_i$ as $v_i$ ranges over current set of vertices, $S \setminus \widehat S$.  Since $c'= v-p'$, $p'=\sum_{j=1}^N \alpha_j \widehat v_j$, where $N= |\widehat S|$,  the evaluation of $c'^T v_i$ requires the computation of $v^T v_i$, and $v_i^T \widehat v_j$, $j=1, \dots, N$.  This requires $O(Nm)$ operations.   Since such computation is only required of each vertex in $\overline S$, over all the computation of all $c'^Tv_i$ requires $O((n-K)m K)=O(nmK)$ operations.  These together with Theorem \ref{thm6}  imply that the over all complexity is $O(mK^2+ nmK + n K/ \gamma^2)$ which is the claimed complexities in (1).

(2): When only $K$ is known, we execute AVTA,  first selecting $\gamma=.5$.  If we compute $K$ vertices with this estimate of $\gamma_*=\Gamma_*/R$, we stop. Otherwise, we halve $\gamma$ and repeat the process.  Eventually in $O(\log (\gamma_*^{-1}))$ calls to AVTA we accumulate all $K$ vertices in $\overline S$.

(3):  For each input $t \in (0,1)$, AVTA computes  a subset  $\overline S^t$ of $\overline S$ with $K^{(t)}$ elements.  The proof of complexity is analogous to the previous cases. Next we prove for each $p \in conv(S)$, the distance from $p$ to $conv( \overline S)$ is at most $tR$. We have
\begin{equation}p= \sum_{i=1}^K \alpha_i \overline v_i, \quad \sum_{i=1}^K \alpha_i =1, \quad \alpha_i \geq 0.\end{equation}
For each $i$  let $p_i \in conv(S^t)$ be the closest point to $\overline v_i$.  Now consider
\begin{equation}p^t= \sum_{i=1}^K \alpha_i p_i.\end{equation}
Then $p^t \in conv(S^t)$.  On the other hand, by the triangle inequality
\begin{equation}\Vert p- p^t \Vert  \leq \sum_{i=1}^K  \alpha_i \Vert \overline v_i - p_i \Vert \leq tR \sum_{i-1}^K \alpha_i= tR.\end{equation}

\end{proof}

\begin{remark}  If nether $K$ nor an estimate $\gamma$ to $\gamma_*=\Gamma_*/R$ are known, initially we select $t=0.5$ and with this value of $t$ compute a subset of vertices with $K^{(t)}$ elements.  We can then halve $t$ and repeat the process.  Intuitively, if for two consecutive values of $t$ no more vertices are generated we can terminate the process, or decrease $t$ by a factor of four. If $\Gamma_*$ is not too small we will produce a reasonably good subset of $\overline S$ within a reasonable number of calls to AVTA. In either case we are assured of an approximation of $conv(S)$ according to (3) in Theorem \ref{thm9}.
\end{remark}

\subsection{Application of AVTA in Solving the Convex Hull Membership}

Suppose we wish to solve the convex hull membership problem: Test if a particular point $p$ lies in $conv(S)$, $S=\{v_1, \dots, v_n\}$.  This is equivalent to linear programming and thus can be solved with variety of algorithms, including polynomial-time algorithms, the simplex method, Frank-Wolfe, or triangle Algorithm.  Whichever algorithm we use, the number $n$ plays a role in the complexity.  Thus if we compute the set of vertices of $conv(S)$, $\overline S$,  we can then test if $p$ lies in $conv(\overline S)$ with $K$ instead of $n$. This approach may seem to be inefficient, however depending upon the accuracy to which we wish to solve the problem and the size of $\gamma_*$ it may result in a more efficient algorithm.  The next theorem considers the application of Theorem \ref{thm9} in solving the convex hull membership problem.

\begin{thm}  \label{thm10} Let $S= \{v_1, \dots, v_n\} \subset \mathbb{R} ^m$.  Let $R$ be the diameter of $S$. Let $ \overline S =\{\overline v_1, \dots, \overline v_K\}$ be the set of vertices of $conv(S)$. Suppose $conv(S)$ is $\Gamma_*$-robust.   Given  any  $0 < \gamma \leq \Gamma_*/R$, the number of operations to test if for a given $p \in \mathbb{R}^m$ admits an $\varepsilon$-approximate solution is
\begin{equation}O \bigg (nmK+ \frac{nK}{\gamma^2} + \frac{K}{\varepsilon^2} \bigg ).\end{equation}
\end{thm}

\begin{proof}   To test if $p$ admits an $\varepsilon$-approximate solution can be achieved by first computing the vertices in $S$, followed by testing if $p$ admits an $\varepsilon$-approximate solution in $conv(S)$.  From  Theorem \ref{thm8} and Theorem \ref{thm3} it follows that the total complexity is as claimed.
\end{proof}

\begin{remark}  It is easy to check that for some values of $\varepsilon < \gamma$ the computations of $\overline S$  followed by testing if $p$ lies in $conv(\overline S)$ could be more efficient than solving the convex hull membership without computing $\overline S$. This is especially true when $K= o(n)$.
\end{remark}

\section{AVTA Under Input Perturbation} \label{AVTAPER}
As in the previous section, we assume $S= \{v_1, \dots, v_n\} \subset \mathbb{R} ^m$, $R$ the diameter of $S$,  and
$\overline S =\{ \overline v_1, \dots, \overline v_K\}$  the set of vertices of $conv(S)$. Assume  $conv(S)$ is $\Gamma_*$-robust.

As before we wish to compute $\overline S$ or a reasonable subset of it. However, in practice the input set $S$ may be not $S$ but a perturbation of $S$.  This changes the set of vertices, robustness parameter and more.  We wish to study perturbations under which we can recover the corresponding perturbation of $\overline S$ and extend AVTA to computing this perturbation.

\begin{definition}
For a given $\varepsilon \in (0,1)$  the $\varepsilon$-perturbations of $S$ is the set $S_\varepsilon$ defined as
\begin{equation}S_\varepsilon =\{v^{\varepsilon}_1, \dots, v^{\varepsilon}_n\},    \quad \Vert v_i -  v^{\varepsilon}_i \Vert \leq \varepsilon R.\end{equation}
The $\varepsilon$-perturbations of $\overline S$ is the set $\overline S_\varepsilon$, denoted by
\begin{equation}\overline S_\varepsilon = \{ \overline v_1^{\varepsilon}, \dots, \overline v_K^{\varepsilon}\},\end{equation}
where $\overline v^{\varepsilon}_i$ is the perturbation of $\overline v_i$.
\end{definition}

In practice we may be given $S_\varepsilon$ as opposed to $S$. The first question that arises is: What is the relationship between the vertices of $S$ and those of $S_\varepsilon$?
Without any assumptions, the vertices of $conv(S_\varepsilon)$ could change drastically, even under small perturbations.

\begin{example} Consider a triangle with three additional interior points, very close to its vertices. It may be the case that even under small perturbation all six points become vertices, or that the interior points become the new vertices while the vertices become the new interior points.  Thus there is a need to make some assumptions before we can say anything about the nature of perturbed points.
\end{example}

We would hope that for appropriate range of values of $\varepsilon$,  $\overline S_\varepsilon$ would at least be a subset of the set of vertices of $S_\varepsilon$.  First we need a definition.

\begin{definition}  \label{defAAB} We say $conv(S)$ is $\Sigma_*$-{\it weakly robust} if
\begin{equation}\Sigma_*= \min \{ d(v,  conv(S \setminus \{ v\})): v \in \overline S \}.\end{equation}
\end{definition}

\begin{example}  Suppose that $S$ consists of the vertices of a non-degenerate triangle with vertices $v_1, v_2, v_3$. Suppose one additional point is placed inside the triangle.  Then clearly $\Sigma_* < \Gamma_*$.
\end{example}

More generally we have

\begin{prop}  \label{prop5} Given $S=\{v_1, \dots, v_n\}$, we have
\begin{equation}\Sigma_* \leq \Gamma_*.  \qed\end{equation}
\end{prop}

Other than the inequality in Proposition \ref{prop5}, $\Sigma_*$ and $\Gamma_*$ corresponding to the set $S$ may seem unrelated, however in the following theorem we establish a relationship between the two that is useful in the analysis of AVTA for computing $\overline S_\varepsilon$.

\begin{thm} \label{thm11} Let $S$ and $\overline S$ be as before. Suppose $conv(S)$ is $\Gamma_*$-robust, also  $\Sigma_*$-weakly robust. Let $\rho_*=\min\{d(v_i, v_j): v_i, v_j \in S, i \not =j\}$.
We have
\begin{equation}  \label{rhogamma}
\Sigma_*  \geq  \frac{\rho_* }{R} \Gamma_* = \rho_* \gamma_*.
\end{equation}
\end{thm}

\begin{proof}  For each vertex $v \in conv(S)$, let $\Gamma_v$ be the distance from $v$ to the convex hull of the remaining vertices in $S$.  Specifically,

\begin{equation}
\Gamma_v= d(v, conv(\overline S \setminus \{v\})).
\end{equation}

Also let $\Sigma_v$ be the distance from $v$ to the convex hull of all other points in $S$. Specifically,

\begin{equation}
\Sigma_v= d(v, conv(S \setminus \{v\})).
\end{equation}
Clearly we have,

\begin{equation}
\Sigma_v \leq \Gamma_v.
\end{equation}

Assume $v$ is a vertex for which $\Sigma_v < \Gamma_v$.  If no such a vertex exists then $\Sigma_* =\Gamma_*$ (see Figure \ref{Fig2}).  Let $u$ be the closest point to $v$  lying in the convex hull of the the other vertices of $S$. Thus
\begin{equation} \label{eqt11a}
\Gamma_v= d(v, u), \quad u \in conv(\overline S).
\end{equation}
Let  $H_u$ be the hyperplane orthogonal to the line segment $vu$, passing through $u$.  By definition of $u$ and
Carath\'eodorey's theorem $u$ is a convex combination of vertices of $conv(S)$ lying on $H_u$.
Thus for some subset $T$ of $ \overline S$
\begin{equation} \label{eqt11b}
u = \sum_{ \overline v_i \in T \subset \overline S} \alpha_i \overline v_i, \quad \sum_{ \overline v_i \in T \subset \overline S} \alpha_i =1,  \quad \alpha_i \geq 0.
\end{equation}

Figure \ref{Fig2} gives a depiction of this property for a simple example.  In the example $u$ is a convex combination of $\overline v$ and $\overline v'$,  vertices of $conv(S)$ lying in the intersection of $H_u$ and $conv(S)$.
Consider one of these vertices, say $\overline v$. Moving the hyperplane $H_u$ parallel to itself toward $v$, it intersects the line segment $uv$ at a unique point $w$ that lies on a facet of $conv(S \setminus \{v\})$.  Such $w$ exists because $\Sigma_* < \Gamma_*$. In other words,  if $H_w$ is a hyperplane parallel to $H_u$ passing through $w$, then the region of $conv(S)$ enclosed between the halfspace defined by $H_w$ and $v$ contains no point of $S$ in its interior (see shaded area in Figure \ref{Fig2}.
This implies
\begin{equation} \label{eqt11c}
\Sigma_v \geq d(v,w).
\end{equation}

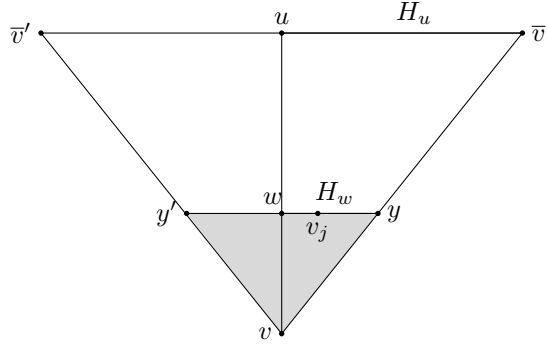
\begin{figure}[htpb]

	\centering
	\begin{tikzpicture}[scale=0.4]	
\begin{scope}[red]
          \clip (-.8,-6) -- (-7.15,-6)  -- (-4,-10) -- cycle;
\fill[color=gray!30] (-10, -10) rectangle (20, 20);
\end{scope}
       \filldraw (-4,0) circle (2pt);
       \filldraw (4,0) circle (2pt);
       \filldraw (-12,0) circle (2pt);
       \filldraw (-4,-10) circle (2pt);
        \draw (-12,0) -- (4,0);
         \draw (-12,0) -- (-4,-10);
      \draw (-4,0) -- (4,0) node[pos=0.55, above] {$H_u$};
      \draw (4,0) -- (-4,-10) node[pos=0.55, above] {};
        \draw (-4,0) -- (-4,-10) node[pos=0.55, above] {};
        \draw (-4,-10) node[left] {$v$};
		\draw (4,0) node[right] {$\overline v$};
        \draw (-12,0) node[left] {$\overline v'$};
        \draw (-4,0) node[above] {$u$};
        \draw (-4,-6) node[above] {$w~~$};
         \draw (-.8,-6) node[right] {$y$};
         \draw (-7.15,-6) node[left] {$y'$};
         \filldraw (-7.16,-6) circle (2pt);
         \draw (-7.15,-6) -- (-.8,-6);
        \filldraw (-4,-6) circle (2pt);
         \filldraw (-.8,-6) circle (2pt);
          \draw (-2.8,-6) node[below] {$v_j$};
           \filldraw (-2.8,-6) circle (2pt);
        \draw (-4,-6) -- (-.8,-6) node[pos=0.55, above] {$H_w$};
	\end{tikzpicture}
\begin{center}
\caption{Given $v \in \overline S$, $u$ is its closet point in $conv(\overline S \setminus \{v\})$. $\overline v, \overline v'\in \overline S$ are vertices of $conv(S)$ lying on $H_u$, the orthogonal hyperplane to line segment $uv$ at $u$. $u$ is a convex combination of these vertices. Moving $H_u$ parallel to itself toward $v$, it intersects the line segment $uv$ at a unique point $w$ lying on a facet of $conv(S \setminus \{v\})$. Thus interior of shaded region contains no point of $S$.} \label{Fig2}
\end{center}
\end{figure}

Now consider the intersection of $H_w$ and each ray connecting $v$ to $\overline v_i \in T$. Denote this intersection by $y_i$.  In the figure the intersection of $H_w$ and the ray connecting $v \overline v$ is denoted by $y$.
By definition of $w$ and Carath\'eodorey's theorem there must exist a point $v_j \in S$ lying on $H_w$. Furthermore, $v_j$ can be written as a convex combination of all the $y_i$'s.  Thus may may write
\begin{equation} \label{eqt11d}
v_j = \sum_{\overline v_i \in T \subset \overline S} \beta_i  y_i, \quad \sum_{\overline v_i \in T \subset \overline S} \beta_i=1, \quad \beta_i \geq 0.
\end{equation}
Since by definition of $\rho_*$, $d(v,v_j) \geq \rho_*$, at least for one $y_i$ we must have $d(v, y_i) \geq \rho_*$.
This implies we could assume $\overline v$ was chosen so that the corresponding $y$ satisfies
\begin{equation} \label{eqt11e}
d(v, y) \geq \rho_*.
\end{equation}

From similarity of the triangles $\triangle vu \overline v$ and  $\triangle vwy$ we may write

\begin{equation} \label{eqt11f}
\frac{d(v,w)}{\Gamma_v} =\frac{d(v, y)}{d(v, \overline v)}.
\end{equation}
From the definition of $R$ as the diameter of $S$,
$d(v, \overline v)  \leq R$.  From (\ref{eqt11f}), (\ref{eqt11c}) and  (\ref{eqt11e}) it  follows that
\begin{equation}
\Sigma_v \geq d(v,w) \geq \frac {1}{R}  d(v, y) {\Gamma_v} \geq \frac{1}{R} \rho_* \Gamma_*.
\end{equation}
This means we have
\begin{equation} \label{eqt11h}
\Sigma_*  \geq \frac{1}{R} \rho_* \Gamma_*.
\end{equation}
\end{proof}

In what follows we will derive complexity bounds for computing $\overline S_\varepsilon$. These complexities will in particular depend on $\Sigma_*$ or any lower bound  $\sigma$ on $\sigma_*=\Sigma_*/R$.  Theorem \ref{thm10} implies that we can choose $\sigma =\rho_* \Gamma_*/R$.

The following theorem describes a simple condition under which the set of vertices of $conv(S)$ under perturbation remain to be vertices of the perturbed convex hull.

\begin{thm} \label{thm12} Let $S$ be as before, $R$ diameter of $S$. Suppose $conv(S)$ is  $\Sigma_*$-weakly robust.  Suppose $S_\varepsilon$ is an $\varepsilon$-perturbation of $S$.  Let $\sigma$ be a positive number satisfying $\sigma \leq \sigma_*=\Sigma_*/R$. Assume
$\varepsilon < \sigma/2$.  If $v \in S$ is a vertex $conv(S)$ and $v^\varepsilon \in S_\varepsilon$  its corresponding $\varepsilon$-perturbation, then $v^{\varepsilon}$ is a vertex of $conv(S_\varepsilon)$.
\end{thm}

\begin{proof} Suppose $v^{\varepsilon}$ is not a vertex of $conv(S_\varepsilon)$.  Without loss of generality assume $v=v_1$. Hence, $v^{\varepsilon}=v^{\varepsilon}_1$.  Thus $v^{\varepsilon} \in conv(S_\varepsilon \setminus \{v^{\varepsilon}\})$.  We may write
\begin{equation}
v^{\varepsilon}=\sum_{i=2}^n \alpha_i v^{\varepsilon}_i, \quad \sum_{i=2}^n \alpha_i=1, \quad \alpha_i \geq 0.
\end{equation}
Set
\begin{equation}  \label{eqfu}
u =\sum_{i=2}^n \alpha_i v_i.
\end{equation}
On the one hand we have
\begin{equation}
u- v^{\varepsilon}  = \sum_{i=2}^n \alpha_i (v_i- v^{\varepsilon}_i).
\end{equation}
Then by the triangle inequality
\begin{equation} \label{equv}
\Vert u- v^{\varepsilon} \Vert  \leq  \sum_{i=2}^n \alpha_i \Vert v_i- v^{\varepsilon}_i \Vert \leq \sum_{i=2}^n \alpha_i {\varepsilon} R ={\varepsilon} R.
\end{equation}
On the other hand,   $v$ is in $\overline S$. Without loss of generality assume $v= \overline v_1$. From this assumption and since by (\ref{eqfu}) $u \in conv(\overline S \setminus \{ \overline v_1\})$ we have
\begin{equation}
u= \sum_{i=2}^K \gamma_i \overline v_i, \quad \sum_{i=2}^K \gamma_i =1, \quad \gamma_i \geq 0.
\end{equation}
Since $conv(S)$ is  $\Sigma_*$-weakly robust on $\overline S$ and $\sigma \leq \sigma_*= \Sigma_*/R$ we  have,
\begin{equation}
\Vert u-v\Vert \geq \sigma R.
\end{equation}
However, from (\ref{equv}), the fact that $\Vert v -  v^{\varepsilon} \Vert  \leq \varepsilon R$ and the triangle inequality we may write.

\begin{equation}
\Vert u - v \Vert = \Vert u - v^{\varepsilon} + v^{\varepsilon} - v \Vert \leq \Vert u - v^{\varepsilon} \Vert  + \Vert v^{\varepsilon} - v \Vert  \leq \varepsilon R + \varepsilon R = 2 \varepsilon R.
\end{equation}
This contradicts the assumption that $2\varepsilon < \sigma$. Hence $v^{\varepsilon}$ is a vertex of $conv(S_\varepsilon)$.
\end{proof}

\begin{remark}  The theorem implies that if  the input to AVTA is $S_\varepsilon$ instead of $S$,  AVTA will still return at least $K$ vertices. However, the set of vertices of $conv(S_\varepsilon)$ may have more elements than $K$, possibly all of $S_\varepsilon$. Moreover, the weakly robustness parameter $\Sigma_*$  will change.  We thus need to revise AVTA if we wish to extract the subset  $\overline S_\varepsilon =\{ \overline v^{\varepsilon}_1, \dots, \overline v^{\varepsilon}_K\}$ from the set of vertices of $conv(S_\varepsilon)$.
\end{remark}

In what follows we will first show how under a mild assumptions on the relationship between  $\Sigma_*/R$ and $\varepsilon$, AVTA can compute  a subset $\widehat S_\varepsilon$ of the vertices of $conv(S_\varepsilon)$ containing $\overline S_\varepsilon$ (Theorem \ref{thm13}). We then show how AVTA can efficiently extract from $\widehat S_\varepsilon$ the desired set, namely $\overline S_\varepsilon$. The next lemma establishes a lower bound on the week-robustness of $conv(S_\varepsilon)$.  It also shows how spurious vertices of $conv(S_\varepsilon)$ are situated with respect to the convex hull of the remaining vertices.   This will be used in Theorem \ref{thm13} in pruning such vertices.

\begin{lemma}  \label{lem2} Suppose $conv(S)$ is  $\Sigma_*$-weakly robust.
	Suppose  $\varepsilon < \Sigma_*/2R$. Let $v^\varepsilon$ be any point in $\overline S_\varepsilon$.
	Let $\widehat S_\varepsilon)$ be any subset of vertices of $conv(S_\varepsilon)$ containing $\overline S_\varepsilon)$. Then,
	\begin{equation} \label{eq1lem2}
	d(v^\varepsilon,  conv(\widehat S_\varepsilon))  \geq  (\Sigma_*-2 \varepsilon R).
	\end{equation}
	Moreover let $\widehat v^\varepsilon$ be any (spurious) point in $\widehat S_\varepsilon \setminus \overline S_\varepsilon$. Then
	\begin{equation} \label{eq2lem2}
	d(\widehat v^\varepsilon,  conv(\widehat S_\varepsilon \setminus \{\widehat v^\varepsilon\}))  \leq   \varepsilon R.
	\end{equation}
	
\end{lemma}

\begin{proof} By Theorem \ref{thm12}, $\overline S_\varepsilon$ is a subset of vertices of $conv(S_\varepsilon)$.   Given $v^{\varepsilon} \in \overline S_\varepsilon$,  let $v$ be the corresponding vertex in $\overline S$.
	Given $w^{\varepsilon}$ in $conv(S_\varepsilon \setminus \{v^{\varepsilon}\})$, let $w$ in  $conv(S \setminus \{v\})$ be the corresponding point, i.e. defined with respect to the same convex combination of corresponding vertices. Then
	\begin{equation}
	\Vert v - v^{\varepsilon} \Vert \leq \varepsilon R, \quad \Vert w - w^{\varepsilon} \Vert \leq \varepsilon R.
	\end{equation}
	From the above it is easy to show
	\begin{equation}
	|d(v, w)- d(v^{\varepsilon}, w^{\varepsilon})  | \leq 2\varepsilon R.
	\end{equation}
	But this implies
	\begin{equation}
	d(v, w) - d(v^{\varepsilon}, w^{\varepsilon}) \leq 2\varepsilon R.
	\end{equation}
	Equivalently,
	\begin{equation}
	d(v, w) - 2\varepsilon R \leq d(v^{\varepsilon}, w^{\varepsilon}).
	\end{equation}
	But $d(v,w) \geq \sigma_* R= \Sigma_*$.  This proves (\ref{eq1lem2}).
	
	To prove (\ref{eq2lem2}),  let $\widehat v$ be the point in $S$ corresponding to $\widehat v^\varepsilon$.  We have
	\begin{equation} \widehat v = \sum_{i=1}^K \alpha_i  \overline v_i,  \quad \sum_{i=1}^K \alpha_i =1, \quad \alpha_i \geq 0.\end{equation}
	Define
	
	\begin{equation} \widehat w = \sum_{i=1}^K \alpha_i  \overline v^\varepsilon_i,  \quad \sum_{i=1}^K \alpha_i =1, \quad \alpha_i \geq 0.\end{equation}
	
	It is now easy to show
	\begin{equation} \Vert \widehat v^\varepsilon - \widehat w \Vert  \leq \Vert \widehat v^\varepsilon - \widehat v \Vert +\Vert \widehat v - \widehat w \Vert \leq 2 \varepsilon R.\end{equation}
	
	This proves (\ref{eq2lem2}).
\end{proof}

\begin{thm}  \label{thm13} Let $S= \{v_1, \dots, v_n \} \subset \mathbb{R}^m$. Assume $conv(S)$ is  $\Sigma_*$-weakly robust.    Suppose $\varepsilon \leq \Sigma_*/ 4R$.

 (i) Given $\sigma$ satisfying, $4\varepsilon \leq \sigma \leq \sigma_*=\Sigma_*/R$,  AVTA can be modified to compute a subset  $\widehat S_\varepsilon$ of the set of vertices of $S_\varepsilon$
containing $\overline S_\varepsilon$, then compute from this subset  $\overline S_\varepsilon$ itself.
If $K_\varepsilon$ is the cardinality of  $\widehat S_\varepsilon$, the total number of operations satisfies
\begin{equation}O \bigg (nK_\varepsilon(m+ \frac{1}{\sigma^{2}}) \bigg ).\end{equation}

(ii) Given $\gamma$, satisfying $4\varepsilon \leq \gamma \rho_* \leq \Gamma_* \rho_*/R= \gamma_* \rho_*$,
AVTA can be modified to compute a subset  $\widehat S_\varepsilon$ of the set of vertices of $S_\varepsilon$
containing $\overline S_\varepsilon$, then compute from this subset  $\overline S_\varepsilon$ itself.
If $K_\varepsilon$ is the cardinality of  $\widehat S_\varepsilon$ the total number of operations satisfies
\begin{equation}O \bigg (n K_\varepsilon (m+ \frac{1}{ (\rho_* \gamma)^{2}} )\bigg ).\end{equation}

(iii) Given only $K$, where $4 \varepsilon \leq \Sigma_*/R$, the number of operations of $AVTA$ to  computes $\overline S_\varepsilon$ is.
\begin{equation}O(nK_\varepsilon(m+  \frac{1}{\sigma_*^{2}})) \log (\frac{1}{\sigma_*}).\end{equation}

(iv) More generally, given any $t \in (0,1)$, AVTA can be modified to compute a subset  $\overline S_\varepsilon^t$ of the set of vertices of $conv( S_\varepsilon)$ of cardinality $K^{(t)}_\varepsilon$ so that the distance from each point in $conv(S_\varepsilon)$  to $conv(\overline S^t_\varepsilon)$ is at most $t$.  In particular, the distance from each point in $conv(S)$  to $conv(S^t_\varepsilon)$ is at most $(t+ \varepsilon) R$. The complexity of the computation of $\overline S_\varepsilon^t$ is
\begin{equation}O \bigg (nK^{(t)}_\varepsilon (m+ \frac{1}{t^2} )\bigg ).\end{equation}

\end{thm}

\begin{proof}  By Theorem \ref{thm12}, $\overline S_\varepsilon$ is a subset of vertices of $conv(S_\varepsilon)$.
Let $\sigma_\circ= (\Sigma_* - 2 \varepsilon R)/R$.  Then since $\varepsilon \leq \Sigma_*/4R$, $\sigma_\circ \geq \Sigma_*/2R$.
Then by Lemma \ref{lem2}, for each $v^\varepsilon \in \overline S_\varepsilon$, we have
\begin{equation}d(v^\varepsilon,  conv(S^\varepsilon \setminus \{ v^\varepsilon \}) \geq \Sigma_*/2R.\end{equation}
Now consider a modification of AVTA that replaces $\gamma/2$, by  $\sigma/2$. Such modified AVTA will compute a subset $\overline S_\varepsilon$  of vertices of $conv(S_\varepsilon)$ that must necessarily contain $\overline S_\varepsilon$. Analogous to  Theorem \ref{thm9}, (1),  the complexity of this part is as stated in part (i) of the present theorem.

Now consider $conv(\widehat S_\varepsilon)$ and assume $v^\varepsilon$ is a vertex of it within a distance of less than $\sigma/ 2$, say $\sigma/4$. Then by Lemma \ref{lem2}, $v^\varepsilon \not \in \overline S^\varepsilon$.
We can thus apply the Triangle Algorithm to remove any vertex of
$conv(\widehat  S_\varepsilon)$ that is within a distance of less than  $\sigma/2$ of the convex hull of the other vertices in $conv(\widehat  S_\varepsilon)$. Again analogous to Theorem \ref{thm9} the over all complexity of this step  is bounded by

\begin{equation}O \bigg (m  K_\varepsilon^2+ \frac{K_\varepsilon^2}{\sigma_\circ^{2}} \bigg )= O \bigg (m  K_\varepsilon^2+ \frac{K_\varepsilon^2}{\sigma^{2}} \bigg ).\end{equation}
This is dominated by the complexity of the first part. This proves (i).   Proof of (ii) follows from Theorem \ref{thm11},   (\ref{rhogamma}), that $\gamma \rho_* \leq \sigma_*$.

To prove (iii), we start by $\sigma=1/2$ and run AVTA. Then as previous case prune unwanted vertices. If we end up with $\overline S_\varepsilon$, we are done. If not, we repeat the process with $\sigma=1/4$ and so on. Eventually we will recover $\overline S_\varepsilon$.

The proof of (iv) is analogous to the proof of Theorem \ref{thm12}, part (3).

\end{proof}

\begin{remark}   Ideally, $K_\varepsilon$ is  within a constant multiple of $K$, in which case the complexities are analogous to those of Theorem \ref{thm9}. In the worst-case  $\widehat S_\varepsilon = S_\varepsilon$, i.e.
$K_\varepsilon = n$.    On the other hand, ignoring the size of $K_\varepsilon$, suppose $\sigma_\circ  \geq (\sqrt{n/K}) \varepsilon$, then the complexity of generating the vertices of $conv(S_\varepsilon)$  is
\begin{equation}O \bigg (nm K_\varepsilon+ \frac{nK}{{\varepsilon}^2} \bigg ).\end{equation}
\end{remark}

\section{Triangle Algorithm with Johnson-Lindenstrauss Projections} \label{sec6}

Consider again $S= \{v_1, \dots, v_n\} \subset \mathbb{R} ^m$.  We wish to compute the subset $\overline S=\{\overline v_1, \dots, \overline v_K\}$ of all vertices of $conv(S)$. Johnson-Lindenstrauss lemma allows embedding the $n$ points of $S$ in an $m'$-dimensional Euclidean space, where $\mathbb{R}^{m'}$, $m' < m$, via a randomized linear map so that the distances between every pair of points in $S$ and those of their images in $\mathbb{R}^{m'}$ remain approximately the same, with high probability.   More specifically,
there is a universal constant $c$ such if $\varepsilon'$ satisfies,
\begin{equation} \label{eqA4}
 \frac{c \log n}{m} \leq \varepsilon'^2 < 1,
\end{equation}
and $m' < m$ is an integer satisfying
\begin{equation} \label{eqA2}
m' \approx \frac{c \log n}{\varepsilon'^2},
\end{equation}
then there exists a randomized linear map $L:  \mathbb{R}^m \rightarrow \mathbb{R}^{m'}$
so that if $u_i=L(v_i)$, and
\begin{equation}
U=L(S)=\{u_1, \dots, u_n\} \subset \mathbb{R}^{m'},
\end{equation}
then for for each $i, j \in \{1, \dots, n\}$ we have
\begin{equation} \label{eqA1}
{\rm Pr} \bigg (d(v_i,v_j) (1- \varepsilon') \leq d(u_i,u_j) \leq d(v_i,v_j) (1+ \varepsilon') \bigg)  > 1 - \frac{2}{n}.
\end{equation}
The projection of each point takes $O(m \log n)$ operations so that the overall number of operations to project all the $n$  points is

\begin{equation}\label{eqA3}
O(nm \log n).
\end{equation}

In this section we consider computing $\overline S$, the set of vertices of $conv(S)$ by using   the Johnson-Lindenstrauss projections and then computing the set of vertices of $conv(U)$ via AVTA.   Let  $\overline U$ denote the set of vertices of $conv(U)$ and let its cardinality be $K'$. First we state some properties of $conv(U)$.

\begin{lemma}  \label{lem8} Given $v \in S$, $L:  \mathbb{R}^m \rightarrow \mathbb{R}^{m'}$,  a randomized linear map, suppose $u=L(v)$ is a vertex of $conv(U)$.  Then $v$ is a vertex of $conv(S)$.
\end{lemma}

\begin{proof} Suppose $v$ is not a vertex of $conv(S)$. Then $v=\sum_{i=1}^n \alpha_i v_i$, $\sum_{i=1}^n \alpha_i=1$, $\alpha_i \geq 0$, $i=1, \dots, n$,  with some $0<\alpha_j <1$.  By linearity of $L$ we have
\begin{equation} u=L(v)= \sum_{i=1}^n \alpha_i L(v_i) = \sum_{i=1}^n \alpha_i u_i.\end{equation}
This implies $u$ is not a vertex of $conv(U)$, a contradiction.
\end{proof}

The next theorem gives an estimate of the robustness parameters of $conv(U)$ in terms of those $conv(S)$.

\begin{thm} \label{thm14} Suppose $conv(S)$ is $\Gamma_*$-robust and
 $\Sigma_*$-weakly robust.  Let $U=L(S)$, $L$ a randomized linear map, $L:  \mathbb{R}^m \rightarrow \mathbb{R}^{m'}$. Let $m'$ and $\varepsilon'$ be related as in (\ref{eqA2}). If   $conv(U)$ is $\Gamma_*'$-robust,
 $\Sigma_*'$-weakly robust, then  with probability at least $(1-2/n)$,  we have
\begin{equation}\Gamma_*' \geq \Gamma_* (1- \varepsilon'),  \quad \Sigma_*' \geq \Sigma_* (1- \varepsilon').\end{equation}
\end{thm}

\begin{proof}  Suppose $u$ is a vertex of $conv(U)$ and $\widehat U$ a subset of its vertices not containing $u$.
Let $v$ and  $\widehat S$ be the preimages of $u$ and $\widehat U$ under the linear map $L$. By Lemma \ref{lem8} $v$ and the elements of $\widehat S$ are all vertices of $conv(S)$. From  (\ref{eqA1}) it is easy to argue that with probability at least $(1-2/n)$ we have
\begin{equation}d(u, conv(\widehat U)) \geq d(v, conv(\widehat S))(1- \varepsilon').\end{equation}
The claimed inequalities follow.
\end{proof}

From Theorem \ref{thm14} and Theorem  \ref{thm9} we can state the following:

\begin{thm}  \label{thm15} Given $S= \{v_1, \dots, v_n\} \subset \mathbb{R} ^m$  let $U = L(S) =\{ u_1, \dots, u_n\} \subset \mathbb{R}^{m'}$, $L$ a randomized linear map, $m'$, $\varepsilon'$ as before.  Let $\overline U= \{ \overline u_1, \dots, \overline u_{K_{\varepsilon'}} \}$ be the set of vertices of $conv(U)$.  Suppose  $conv(S)$ is $\Gamma_*$-robust and $conv(U)$ is $\Gamma_*'$-robust. Then with probability at least $(1-2/n)$,

(1) The number of arithmetic operations of AVTA to compute $\overline U$ is
\begin{equation}O \bigg (nK_{\varepsilon'}(m'+ \frac{nK_{\varepsilon'}{R^2}}{\Gamma_*'^{2}}) \bigg )=
O \bigg (\frac{n \log n K_{\varepsilon'}}{\varepsilon'^2}+ \frac{nK_{\varepsilon'}}{\gamma_*^{2}(1 -\varepsilon')^2} \bigg ).\end{equation}

(2)  Given any prescribed positive $t \in (0,1)$, AVTA in
\begin{equation}O \bigg (nK_{\varepsilon'}^t(m'+ \frac{1}{t^2}) \bigg )=
O \bigg ( \frac{n \log nK_{\varepsilon'}^t}{\varepsilon'^2}+ \frac{nK_{\varepsilon'}^t}{t^2} \bigg )\end{equation}
operations can compute a subset $\overline U^t$ of $\overline U$ of size $K_{\varepsilon'}^t$ so that the distance from each point in $conv(U)$  to $conv(\overline  U^t)$ is at most $t$. $\Box$
\end{thm}

\begin{remark}  The results in this section and the above theorem suggest a heuristic  approach as an alternative to using AVTA directly to compute all the vertices of $conv(S)$: Compute $U=L(S)$, the Johnson-Lindenstrauss projection of $S$ under a randomized linear map $L$.  Then apply AVTA to compute all the vertices of $conv(U)$, $\overline U$. This identifies $|\overline U| \leq K$ vertices of $conv(S)$.  Next move up to the full dimension and continue with AVTA to recover the remaining vertices of $conv(S)$.  Alternatively, we can repeat randomized projections and compute the corresponding vertices.  We would have to delete duplications which is not difficult, given that we store the computed vertices via their vector of representation of convex combination coefficients.
We would expect that when sufficient number of projections are applied all vertices of $conv(S)$ can be recovered.
However,  in the remaining of the section we analyze the probability that under a random projection, the projection of a vertex of $conv(S)$ is a vertex of the projection.
\end{remark}
\noindent
In what follows we first state a result on Johnson-Lindenstrauss random projections on the convex hull membership problem from \citet{vu2017random}. Next we state an alternative result.

\begin{prop} \label{propL} {\rm (\citet{vu2017random}, Proposition 3.3)} Given $S= \{v_1, \dots, v_n \} \subset  \mathbb{R}^m$, $p \in \mathbb{R}^m$ such that $p \not \in conv(S)$, let $d= \min\{d(p,x): x \in conv(S)\}$ and $D= \max \{ d(p,v_i): i=1, \dots, n\}$. Let $T: \mathbb{R}^m \rightarrow \mathbb{R}^k$ be a random linear map. Then
\begin{equation}
{\rm Prob} \bigg ( T(p) \not \in T(conv(S)) \bigg ) \geq 1- 2n^2 e^{-c(\varepsilon^2 - \varepsilon^3)k}
\end{equation}
for some constant $c$ (independent of $m,n,k,d,D$) and $\varepsilon < d^2/D^2$. \qed
\end{prop}

\begin{remark}
Note that $k= O(\ln n/ \epsilon^2)= O ( \ln n D^4/d^4)$.
\end{remark}

The following is an alternative to Proposition \ref{propL} based on the Distance Duality theorem (\ref{thm1}) and generally gives a better estimate of $\varepsilon$, hence a smaller  $k$ than Proposition \ref{propL}.

\begin{thm} \label {thmCHM}   Given $S= \{v_1, \dots, v_n \} \subset  \mathbb{R}^m$, $p \in \mathbb{R}^m$ such that $p \not \in conv(S)$, let $d= \min\{d(p,x): x \in conv(S)\}$,  $p_*= {\rm argmin} \{d(p,x): x \in conv(S)\}$ and $D= \max \{ d(p,v_i): i=1, \dots, n\}$. Let
\begin{equation}
E= \min \bigg \{ \frac{d(p, v_i)}{d(p_*, v_i)} : i=1, \dots, n \bigg \}.
\end{equation}
Let $T: \mathbb{R}^m \rightarrow \mathbb{R}^k$ be a random linear map. Then
\begin{equation}
{\rm Prob} \bigg (T(p) \not \in T(conv(S)) \bigg ) \geq 1- 2n^2 e^{-c\varepsilon^2 k},
\end{equation}
for some constant $c$ (independent of $m,n,k,d,D$) and $\varepsilon < (E-1)/(E+1)$.  Furthermore,
${(E-1)}/{(E+1)} >  {d^2}/{4D^2}$.
\end{thm}

\begin{proof} Since $p_*$ is the closest point to $p$ in $conv(S)$, it is easy to show that it is a $p$-witness, i.e.
\begin{equation}
d(p_*, v_i) < d(p,v_i),  \quad \forall i=1, \dots, n.
\end{equation}
Let  $\overline p=T(p)$, $\overline p_*=T(p_*)$, and for $i=1, \dots, n$,  $\overline v_i=T(v_i)$. We now consider the set of $n+1$ points $\{v_0=p, v_1, \dots, v_n\}$ and their random projections and
find condition on $\varepsilon$ such that $\overline p_*$ will be an $\overline p$-pivot with respect to $T(conv(S))$, probabilistically.   By the Johnson-Lindenstrauss Lemma we have,
\begin{equation} \label{eqJLL}
{\rm Prob} \bigg ( (1- \varepsilon) d(v_i, v_j) \leq d(\overline v_i, \overline v_j) \leq (1+ \varepsilon) d(v_i, v_j) \bigg) \geq  1- 2(n+1)^2 e^{-c \varepsilon^2 k},
\end{equation}
for some constant $c$ (independent of $m,n,k$).
From  (\ref{eqJLL}) and definition of $E$, for each $i=1, \dots, n$ with probability at least  $1- 2(n+1)^2 e^{-c \varepsilon^2 k}$ we have,
\begin{equation}
d(\overline p_*, \overline v_i) \leq (1+ \varepsilon) d(p_*, v_i) \leq \frac{(1+ \varepsilon)}{E} d(p, v_i) \leq \frac{(1+ \varepsilon)}{(1- \varepsilon)}  \frac{1}{E} d( \overline p, \overline v_i).
\end{equation}
Note that assuming $n \geq 2$, $ 1 < E < \infty$.  We thus restrict $\varepsilon$ to satisfy
\begin{equation}
\frac{(1+ \varepsilon)}{(1- \varepsilon)}  \frac{1}{E} < 1.
\end{equation}
Equivalently,
\begin{equation}
\varepsilon < \frac{E-1}{E+1}.
\end{equation}
Thus with $\varepsilon$ satisfying the above, $\overline p_*$ is a witness with high probability.

Next we find a lower bound on the right-hand-side of the above. Since $E$ is finite,
$E=d(p, v_j)/d(p_*, v_j)$ for some $j$, i.e. $p_* \not = v_j$.  Consider the triangle with vertices $p$, $v_j$ and $p_*$.   With $d(p,p_*)$ and $d(p,v_j)$ fixed, the maximum value of $d(p_*, v_j)$ is $\sqrt{d^2(p,v_i) - d^2(p,p_*)}$. Using this we may write
\begin{equation} \label{eqE}
E=\frac{d(p, v_j)}{d(p_*, v_j)} \geq \frac{d(p,v_j)}{\sqrt{d^2(p,v_i) - d^2(p,p_*)}} = \frac{1} {\sqrt{1- d^2(p,p_*)/d^2(p,v_j)}}.
\end{equation}
But $d(p,p_*)=d$ and $d(p,v_j) \leq D$. Thus
\begin{equation} \label{eqEE}
E \geq  \frac{1} {\sqrt{1- d^2/D^2}} = \frac{D} {\sqrt{D^2- d^2}}.
\end{equation}
The function $(x-1)/(x+1)$ is  monotonically increasing. Thus from (\ref{eqEE}) we have
\begin{equation} \label{eqEE}
\frac{E-1}{E+1} \geq  \frac{D -\sqrt{D^2- d^2}}{D +\sqrt{D^2+ d^2}} = \frac {d^2}{ (D +\sqrt{D^2- d^2})^2} \geq \frac{d^2}{4D^2}.
\end{equation}
\end{proof}

\begin{remark}
We would expect that $(E-1)/(E+1)$ is generally a larger number than $d^2/4D^2$.
Thus Theorem \ref{thmCHM} gives generally a better estimate of $\varepsilon$ and $k$ than those of Proposition \ref{propL}. An additional advantage of Theorem \ref{thmCHM} is that it shows the applicability of the Triangle Algorithm in solving the convex hull membership problem using random projections.
\end{remark}

We now state a corollary of the theorem on computation of all vertices of $conv(S)$.

\begin{cor} Given $S= \{v_1, \dots, v_n \} \subset \mathbb{R}^m$, suppose  $conv(S)$ is $\Gamma_*$-robust. Let $R$ be the diameter of $S$. Suppose $v_j$ is a vertex of $conv(S)$.
Let $T: \mathbb{R}^m \rightarrow \mathbb{R}^k$ be a random linear map. Then the probability that $T(v_j)$ is a vertex of $T(conv(S))$ is at least $1- 2n^2 e^{-c\varepsilon^2k}$,
for some constant $c$ (independent of $m,n,k$) and $\varepsilon < \gamma_*^2/4$.
\end{cor}

\begin{proof} We apply the previous theorem with $v_j$ as $b$ and considering the probability that under a random projection of $v_j$ lies in projection of the convex hull of the remaining points.
Note that  $d(v_j, conv(S \setminus \{v_j\}) \geq \Gamma_*$ and $\max \{d(v_j, v_i): v_i \in S \setminus \{v_j\} \leq R$.  Thus we can replace for $b/D$ in (\ref{eqEE}) in the previous theorem by $\gamma_*=\Gamma_*/R$.  Thus we can write  $(E-1)/(E+1) \geq {\gamma_*^2}/{4}$. This gives the upper bound on $\varepsilon$.
\end{proof}


\subsection{AVTA Under Perturbation and Johnson-Lindenstrass Projection}

Let $S_\varepsilon$ be as before and
$U_\varepsilon$, a subset of $\mathbb{R}^{m'}$ the  perturbation of $U$.  Let $\overline U_\varepsilon$ be the perturbation of $\overline U$. Based on the results in this section and previous complexity bounds we have

\begin{thm}  \label{thm16} Let $S= \{v_1, \dots, v_n \} \subset \mathbb{R}^m$.  Assume $conv(S)$ is $\Sigma_*$-weakly robust.  Suppose $\varepsilon < \Sigma_*/ 4R$. Let $\sigma_\circ= \overline (\Sigma_* - 2 \varepsilon R)/R= \sigma_* - 2 \varepsilon$.  Then with probability at least $(1-2/n)$,

(i) AVTA can be modified to compute a subset $ \widehat U_\varepsilon$  of $U_\varepsilon$, of cardinality $K_{\varepsilon \varepsilon'}$  such that  it contains $\overline U_\varepsilon$. Then AVTA can compute from this subset  $\overline U_\varepsilon$ itself, where the total number of operations satisfies

\begin{equation}O \bigg (nm'  K_{\varepsilon \varepsilon'}+ \frac{n  K_{\varepsilon \varepsilon'}}{\sigma_\circ^{2}(1 - \varepsilon')^2} \bigg )=
O \bigg (\frac{n \log n  K_{\varepsilon \varepsilon'} }{\varepsilon^2}+ \frac{n K_{\varepsilon \varepsilon'}}{\sigma_\circ^{2}(1 - \varepsilon')^2} \bigg ).\end{equation}

(ii)  Given any prescribed positive $t \in (0,1)$, in
\begin{equation}O \bigg (\frac{n \log n K^{(t)}_{\varepsilon \varepsilon'}}{\varepsilon^2}+ \frac{nK^{(t)}{\varepsilon \varepsilon'}}{\overline \sigma_\circ^2 (1 - \varepsilon')^2}\bigg )\end{equation}
operations the modified  AVTA can compute a subset $U^t_\varepsilon$ of $\overline U_\varepsilon$  of size $K^{(t)}_{\varepsilon \varepsilon'}$ so that the distance from each point in $conv(U_\varepsilon)$  to $conv(U^t_\varepsilon)$ is at most $t$.
\end{thm}

\section{Applications} \label{app}
While the modified AVTA algorithm comes with theoretical guarantees, in certain cases the algorithm might output many more vertices, $K_\varepsilon$, than desired. Here we present a practical implementation that always
outputs exactly $K$ vertices, provided $K$ is known. When $K$ is unknown, our experiments in the next section reveal that the algorithm can automatically detect a slightly larger set that contains a good approximation to the $K$ vertices of interest. Notice that we want a fast way to detect good approximations to the original vertices of the set $S$ and prune out spurious points, i.e., additional vertices of the set $S_\varepsilon$. The key insight on top of the AVTA algorithm is the following: {\em If the perturbed set is randomly projected onto a lower dimensional space, it is more likely for an original vertex to still be a vertex than for a spurious vertex}. Using this insight the algorithm outlined below runs the modified AVTA algorithm over several random projections and outputs the set of points that appear as vertices in many random projections.
\begin{center}

\begin{tikzpicture}
\node [mybox] (box){%
    \begin{minipage}{0.8\textwidth}
{\bf  AVTA with multiple random projections ($S=\{v_1, \dots, v_n\}$, $K$, $\gamma$, $M$)}\
\vspace{.2cm}
\begin{itemize}
\item
{\bf Step 0.} Set $Freq \leftarrow 0^{|S|}$.

\item
{\bf Step 1.} For $i=1$ to $M$:

\begin{itemize}

\item
$S' \leftarrow S $: Project data on to randomly chosen
$\frac{4log(n)}{\epsilon^2}$ dimensions.

\item
$\hat{S} \leftarrow $ AVTA$(S',\gamma)$

\item
For each $d_j \in \hat{S}$, $Freq[j] = Freq[j] + 1$.

\end{itemize}

\item
{\bf Step 2.} Output top $K$ frequent vertices.

\end{itemize}
\end{minipage}};
\label{avtamultiple}
\end{tikzpicture}
\end{center}

We now show how AVTA can be used to solve various problems in computational geometry and machine learning.

\noindent \textbf{Application of AVTA in Linear Programming:}
Consider linear programming feasibility problem of testing if $ P=\{x \in \mathbb{R}^n: Ax=b, x \geq 0\}$
is nonempty,  where $A$ is $m \times n$, $b \in \mathbb{R}^n$. Suppose $n$ is much larger than  $m$.  If we reduce the size of $A$ the problem would be more efficiently solvable, no matter what algorithm we use to solve it.

\begin{prop}  \label{LP1}
Given $ P=\{x \in \mathbb{R}^n: Ax=b, x \geq 0\}$, let $conv(A)$ denote the convex hull of columns of $A$. Let $A'$ denote the $m \times n'$ submatrix $A$ whose columns form the set of all vertices of $conv(A)$. Let
\begin{equation}
P'= \{x' \in \mathbb{R}^{n'}: A'x'=b, x' \geq 0\}.
\end{equation}
Then $P$ is feasible if and only $P'$ is feasible.
\end{prop}

\begin{proof} Clearly, if $P'$ is feasible then $P$ is feasible.  Assume $P$ is feasible. Thus for  some $x \in \mathbb{R}^n$,  $x \geq 0$, $Ax=b$.  Denote the columns of $A$ by $a^{(i)}$.  Then each $a^{(i)}$ is a convex combination of columns of $A'$.  That is, for each $i=1, \dots, n$, there exists
\begin{equation}
\alpha^{(i)} \in S_{n'}= \{s \in \mathbb{R}^{n'}:  \sum_{i=1}^{n'} s_i=1, s \geq 0 \},
\end{equation}
where
\begin{equation}
a^{(i)} = A' \alpha^{(i)}.
\end{equation}
Thus
\begin{equation}
Ax= \sum_{i=1}^n x_i a^{(i)} = \sum_{i=1}^n x_i A' \alpha^{(i)} = A' \sum_{i=1}^n x_i \alpha^{(i)}.
\end{equation}
Letting
\begin{equation}
x'= \sum_{i=1}^n x_i \alpha^{(i)},
\end{equation}
$A'x'=b$, $x' \geq 0$.
\end{proof}

\begin{prop}
Assume  $ P=\{x \in \mathbb{R}^n: Ax=b, x \geq 0\}$ is nonempty.  Consider the linear program
$\min \{c^Tx : x \in P\}$.  Let $B$ be the $(m+1) \times n$ matrix whose first row is $c^T$ and the remaining rows are $A$.   Let $B'$ be the $(m+1) \times n'$ matrix whose columns form the vertices of the convex hull of the columns of $B$.  Let $c'^T$ be the first row of $B'$ and $A'$ the remaining $m \times n'$ submatrix of $B'$.  Then
\begin{equation}
\min \{c^Tx: Ax=b, x \geq 0\} = \min \{c'^Tx: A'x'=b, x' \geq 0\}.
\end{equation}
\end{prop}
\begin{proof}
Consider any feasible solution $x_0$ of original LP.  Then by Proposition \ref{LP1} the set $\{c'^Tx' =c^Tx_0, A'x'=b, x' \geq 0\}$ is feasible.  This implies the original LP has a finite optimal value if and only if the restricted problem does.  In particular, the optimal objective values of the two problems coincide.
\end{proof}

The above propositions imply that AVTA has potential applications in the reduction of the LP feasibility or optimization, whether we solve the problem via simplex method or other methods.

\noindent \textbf{AVTA for topic modeling in the presence of anchor words:}
~\citet{arora2013practical} provide a practical algorithm for topic modeling with provable guarantees. Their algorithm works under the assumptions that the topic-word matrix is {\em separable}. In particular, they assume that corresponding to each topic $i$, there exists an {\em anchor word} $w_i$ that has a non zero probability of appearing only under topic $i$. Under this assumption, the algorithm of~\citet{arora2013practical} consists of two stages: a) find the anchor words, and b) use the anchor words to learn the topic word matrix. The problem of finding anchor words corresponds to finding the vertices of the convex hull of the word-word covariance matrix.  They propose an algorithm named {\em fast anchor words} in order to find the vertices. Since AVTA works in general setting, we can instead use AVTA to find the anchor words. Additionally, the fast anchor words algorithm needs to know the value of the number of anchor words, as an input. On the other hand, from the statements of Theorems~\ref{thm9} and \ref{thm13} it is easy to see that AVTA can work in a variety of settings when other properties of the data are known such as the robustness. We argue that robustness is a parameter that can be tuned in a better manner than trying different values of  the number of anchor words. In fact, one can artificially add random noise to the data and make it robust up to certain value. One can then run AVTA with the lower bound on robustness as input and let the algorithm automatically discover the number of anchor words. This is much more desirable in practical settings. Our first implementation of AVTA is named AVTA+RecoverL2 that uses AVTA to detect anchor words and then uses the anchor words to learn the topic word matrix using the approach from~\citet{arora2013practical}. AVTA is also theoretically superior than fast anchor words and achieves slightly better run times in the regime when the number of topics is $o(\log n)$, where $n$ is the number of words in the vocabulary. This is usually the case in most practical scenarios.

\noindent \textbf{AVTA for topic modeling the absence of anchor words:} The presence of anchor words is a strong assumption that often does not hold in practice. Recently, the work of ~\citet{bansal2014provable} designed a new practical algorithm for topic models under the presence of catch words. Catch words for topic $i$ correspond to set $S_i$ such that it's total probability of appearing under topic $i$ is significantly higher than in any other topic. Their algorithm called TSVD recovers much better reconstruction of the topic-word matrix in terms of the $\ell_1$ error. They also assume that for each topic $i$, there are a few dominant documents that mostly contain words from topic $i$. The TSVD algorithm works in two stages. In stage 1, the (thresholded) word-document data matrix is projected onto a $K$-SVD space to compute a different embedding of the documents. Then, the documents are clustered into $K$ clusters. Under the assumptions mentioned above, one can show that the dominant documents for each topic will be clustered correctly. In stage 2, a simple post processing algorithm can approximate the topic-word matrix from the clustering.

We improve on TSVD by asking the following question: {\it is $K$-SVD the right representation of the data?}. Our key insight is that if dominant documents are present in the topic, it is easy to show that most other documents will be approximated by a convex combination of the dominant topics. Furthermore, the coefficients in the convex combinations will provide a much more faithful low dimensional embedding of the data. Using this insight, we propose a new algorithm that runs AVTA on the data matrix to detect vertices and to approximate each point using a convex combination of the vertices. We then use the coefficient matrix as the new representation of the data that needs to be clustered. Once the clustering is obtained, the same post processing step from~\citet{bansal2014provable} can be used to recover the topic-word matrix. Our results show that the embedding produced by AVTA leads to much better reconstruction error than of that produced by TSVD. Furthermore, $K$-SVD is an expensive procedure and very sensitive to the presence of outliers in the data. In contrast, our new algorithm called AVTA+CatchWord is much more stable to noise in the data.
%
%
%
%
%
%

\begin{center}
	
	\begin{tikzpicture}
	\node [mybox] (box){%
		\begin{minipage}{0.9\textwidth}
		{\bf  AVTA+CatchWord  ($S=\{v_1, \dots, v_n\}$, $\gamma$, $K$, $\epsilon$)}
		\vspace{.2cm}
		\begin{itemize}
		\item
		{\bf Step 0.} Randomly project $S$ onto $2K$ dimensions to get $\hat{S}$.
		
		\item
		{\bf Step 1.} Compute a  a super set of vertices $\bar{V}$ by $AVTA(\hat{S}, \gamma)$.

		\item
		{\bf Step 2.} Prune $\bar{V}$ into $\hat{V}$ (of size $K$) by iteratively picking $\bar{v} \in \{\bar{V}\}/\{\hat{V}\}$ which has the maximum distance to $conv(\hat{V})$.

		\item
		{\bf Step 3.} For each projected point $\hat{v}_i \in \hat{S} \setminus \hat{V}$, compute a vector $\alpha_i$ such that $\|\hat{V} \alpha_i - \hat{v}_i\| \leq \epsilon$.
		
		\item
		{\bf Step 4.} Initialize cluster assignment for each point by majority weight: $ \argmax\limits_{j \in [K]} \alpha_j$	.
		
		\item
		{\bf Step 5.} Clustering using Lloyds algorithm on the embedding provided by the $\alpha$ vectors.
		
		\item
		{\bf Step 6.} Use the post processing as described in~\citet{bansal2014provable} to recover the topic-word matrix from the clustering.
		
		\end{itemize}
		\end{minipage}};
	\end{tikzpicture}
\end{center}

\noindent \textbf{AVTA for NMF:} The work of~\citet{arora2012computing} showed that convex hull detection can be used to solve the non-negative matrix factorization problem under the separability assumption. We show that by using the more general AVTA algorithm for solving the convex hull problem results in comparable performance guarantee.
\section[Applications and Experiments]{Applications and Experiments  ~\footnote{Resources: \url{https://github.com/yikaizhang/AVTA}}}

\subsection{Feasibility problem}

In this section, we present experimental results which empirically show when the problem is 'over complete', AVTA can be a 'shortcut' solution. In another word, given an $m\times n$ matrix $A$ as data, where the convex hull of the columns of $A$, denoted by $conv(A)$, has $K$ vertices, $K \ll n$. We apply the AVTA to solve $2$ classical problems which appear in many applications.

\noindent\textbf{Convex hull membership problem:}

\noindent In the experiments, vertices of the convex hull are generated by the Gaussian distribution, i.e. $v_i \sim \mathcal{N}(0,\mathcal{I}_m), i \in [K]$. Having generated the vertices, the 'redundant' points $d_j$ where $d_j \in conv(S), j\in [n-K]$ are produced using random convex combination $d_j=\sum_{i=1}^{K} \alpha_i v_i$. Here $\alpha_i$ are scaled standard uniform random variable where $\alpha_i$ are scaled so that $\sum_{i=1}^{K} \alpha_i=1$. Specifically, comparison is by fixing $K=100$, $m=50$ and  $n$ varying from $5,000\sim 500,000$.  We compare the efficiency of $4$ algorithms on solving this problem: the Simplex method ~\citet{chvatal1983linear}, the Frank Wolfe Algorithm (FW) ~\citet{jaggi2013revisiting}, the Triangle Algorithm (TA) ~\citet{kalantari2015characterization},  and our algorithm on solving the convex hull membership query problem.

\begin{table}[h]
	\centering
	\caption{Running time of convex hull memberhip (secs)}
	\label{tb:cv_query}
	\scalebox{1}
	{
		\begin{tabular}{|l|l|l|l|l|}
			\hline
			\begin{tabular}[c]{@{}l@{}}\# of \\ redundant pts\end{tabular} & AVTA  & TA & FW & Simplex \\ \hline
			5,000                                                               & 1.75  & 0.21               & 0.52        & 0.9     \\ \hline
			20,000                                                              & 1.49  & 0.66               & 1.94        & 2.76    \\ \hline
			45,000                                                              & 2.94  & 1.84               & 5.51        & 6.16    \\ \hline
			80,000                                                              & 2.71  & 3.22               & 10.87       & 10.63   \\ \hline
			125,000                                                             & 3.83  & 4.28               & 17.67       & 15.95   \\ \hline
			180,000                                                             & 4.15  & 5.38               & 23.14       & 24.13   \\ \hline
			245,000                                                             & 6.95  & 9.56               & 33.42       & 36.96   \\ \hline
			320,000                                                             & 8.09  & 13.24              & 44.99       & 44.26   \\ \hline
			405,000                                                             & 10.01 & 14.75              & 56.35       & 59.5    \\ \hline
			500,000                                                             & 14.12 & 15.69              & 70.7        & 90.41   \\ \hline
		\end{tabular}
	}
\end{table}
\noindent \textbf{Results on Convex hull membership query:} Table ~\ref{tb:cv_query} shows when $n \gg K$, AVTA is more efficient than other algorithms solving the convex hull membership problem. This result supports the output sensitivity property of AVTA. \\

\noindent\textbf{Non-negative linear system:} The non-negative linear system problem is to find a feasible solution of :\\
\begin{equation}
\begin{aligned}
&  A\alpha=p\\
& \alpha\geq0\;\;
\end{aligned}
\end{equation}
In another word, to test if $p \in cone(A_j)$ where $A_j$ are columns of $A$. In case when $A$ is over complete,  any feasible $p$ can be represented using only the generators of $cone(A)$ the set $\bar{A} \subset A$. By scaling $A$ so that columns of $AD\; (D_{ii}= \frac{b}{a\cdot A_i})$ are in a  $m-1$ dimensional hyperlane $\langle a, \alpha\rangle =b$, one can find the generators of $cone(A)$ by finding the vertices of the convex hull of the projected points. This could be done efficiently by AVTA.
Suppose we have a linear system $A$ and series of query points $p$ , it is sufficient to run AVTA once for dimension reduction and solve the subproblem $\bar{A}\alpha'=p, \alpha' \geq0$ using simplex method.
\noindent
We compare the running time of Simplex Method with AVTA+Simplex Method. The generator $\bar{A}$ is entrywise independent $uniform(0,1)$ random matrix and the 'overcomplete' part of the matrix $\bar{A}^c=A/\bar{A}$ are generated by $\bar{A}^c=\bar{A}B$ where $B\in \mathbb{R}^{ K \times (n-K)}$ is entrywise independent $uniform(0,10)$ random matrix. We set the number of generators $K=100$, the dimension $m=50$, and the number of 'redundant' columns $n=50,000$. We simply set half of the query points feasible and rest infeasible.  The  feasible points  $p$ are generated as $p=Ax$ where $x\in \mathbb{R}^{n}$ is entrywise independent $uniform(0,1)$ random vector and the infeasible points are generated in the same way as generators.

\begin{table}[h]
	\centering
	\caption{Running time of linear programming feasibility (secs)}
	\label{tb:nnlsys time}
	\scalebox{1}
	{
		\begin{tabular}{|l|l|l|l|l|l|}
			\hline
			\begin{tabular}[c]{@{}l@{}}\# of\\   query\end{tabular} & AVTA+Simplex & Simplex & \# of query & AVTA+Simplex & Simplex \\ \hline
			1.00                                                    & 241.24       & 152.09  & 11.00       & 241.94       & 1810.72 \\ \hline
			2.00                                                    & 241.36       & 303.86  & 12.00       & 242.01       & 1967.93 \\ \hline
			3.00                                                    & 241.41       & 477.89  & 13.00       & 242.07       & 2125.62 \\ \hline
			4.00                                                    & 241.45       & 660.95  & 14.00       & 242.16       & 2289.91 \\ \hline
			5.00                                                    & 241.54       & 853.91  & 15.00       & 242.23       & 2490.52 \\ \hline
			6.00                                                    & 241.61       & 1016.77 & 16.00       & 242.29       & 2680.61 \\ \hline
			7.00                                                    & 241.69       & 1177.30 & 17.00       & 242.32       & 2866.23 \\ \hline
			8.00                                                    & 241.72       & 1336.38 & 18.00       & 242.41       & 3065.50 \\ \hline
			9.00                                                    & 241.83       & 1495.70 & 19.00       & 242.44       & 3245.78 \\ \hline
			10.00                                                   & 241.91       & 1652.84 & 20.00       & 242.47       & 3412.39 \\ \hline
		\end{tabular}
	}
\end{table}

\begin{figure}
	\centering
	\begin{subfigure}{0.450\textwidth}
		\centering
		\includegraphics[width=0.95\textwidth]{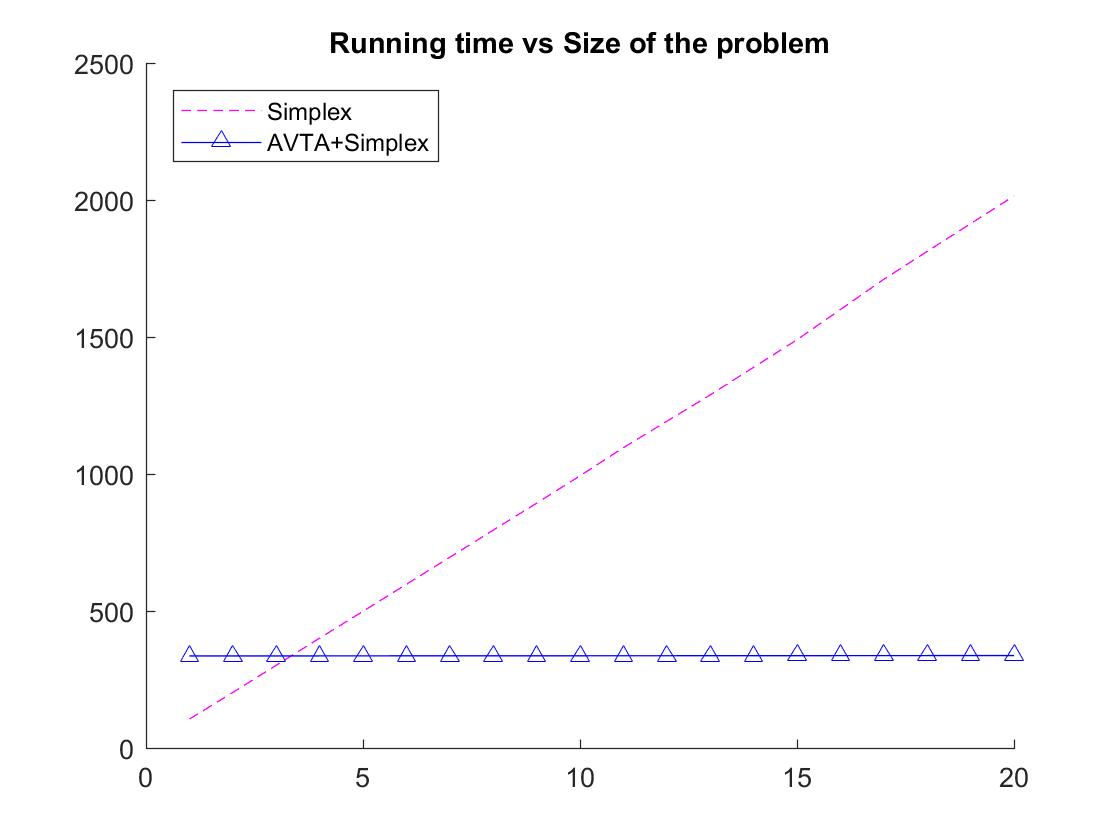}
		\caption{Running time for algorithms to find a feasible solution.}
		\label{fig:nnlsys-time}
	\end{subfigure}
	\begin{subfigure}{0.450\textwidth}
		\centering
		\includegraphics[width=0.95\textwidth]{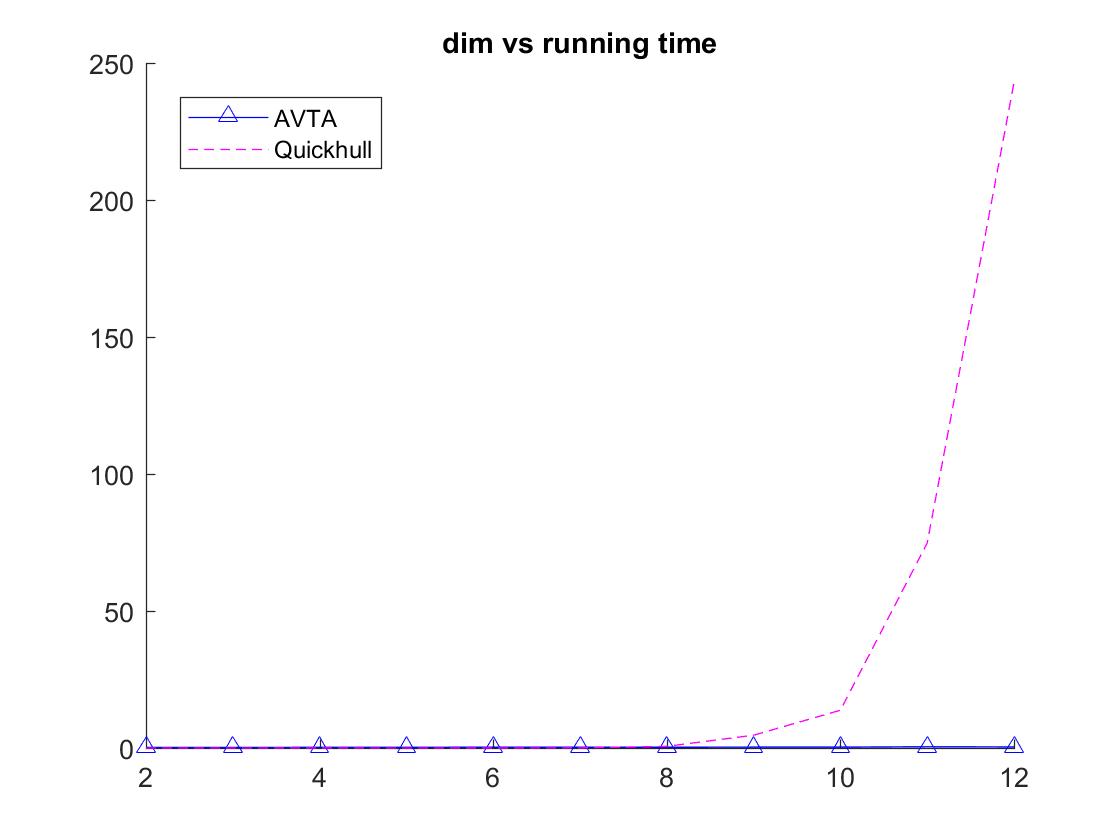}
		\caption{Running time for algorithms to find all vertices.}
		\label{fig1}
	\end{subfigure}
\end{figure}

\noindent It can be observed from Figure~\ref{fig:nnlsys-time} and Table~\ref{tb:nnlsys time} that the running time of AVTA+Simplex doesn't have obvious increase while Simplex increases drastically. This suggests the potential applications of AVTA in linear programming feasibility problem.

\subsection{Computing all vertices}

\noindent\textbf{Compute vertices of convex hull:} In this section, we compare the efficiency of AVTA with another popular algorithm for finding vertices Quickhull ~\citet{barber1996quickhull}. We generate vertices according to a Gaussian distribution $\mathcal{N}(0,10) ^m$. Having generated $K$ such points, $n$ interior points are generated as convex combination of the vertices, where the weights are generated  scaled i.i.d uniform distribution.


\noindent \textbf{Experiment and results:}
In the experiment, we set $K=100$, $n=500$ and $m$ varying from $2 \sim 12$.~\footnote{ The maximum of dimension is $12$ in the experiment because of the explosion of running time of the Quick hull algorithm}.
\noindent
The computational results is shown in Table ~\ref{tb:avtaqhull}.  In high dimension $m \geq 9$, when $conv(S)$ is $\gamma$ robust for some $\gamma >0$, the AVTA algorithm successfully find all vertices of the convex hull efficiently while the Quick hull algorithm is stuck by its explosion of complexity in dimension $m$.

\begin{table}[!h]
	\centering
	\caption{Running time (secs)}
	\label{tb:avtaqhull}
	\scalebox{0.9}
	{
		\begin{tabular}{|l|l|l|l|l|l|}
			\hline
			dim & Qhull & AVTA  & dim & Qhull & AVTA  \\ \hline
			2   & 0.13      & 14.82 & 7   & 2.92      & 41.51 \\ \hline
			3   & 0.02      & 16.62 & 8   & 16.48     & 39.63 \\ \hline
			4   & 0.04      & 24.49 & 9   & 82.09     & 44.21 \\ \hline
			5   & 0.12      & 32.76 & 10  & 391.36    & 45.79 \\ \hline
			6   & 0.59      & 37.66 & 11  & 1479.51   & 51.19 \\ \hline
		\end{tabular}
	}
\end{table}

\noindent\textbf{Compute vertices of simplex in high dimension:} The Fast Anchor Word can be used to detect the vertices of a simplex. In this section, we compare the efficiency of AVTA with Fast Anchor Word when convex hull is a simplex with $K=50$ and  $m=100$. The number of points in the convex hull $n$ varies from $100\sim 100,000 $.

\noindent \textbf{Results of running time in simplex case:} \noindent The running  of  efficiency comparison between AVTA and Fast Anchor Word in simplex case is presented in Figure ~\ref{fig:run_time_avta_faw}. In regime $n\geq 30,000$, AVTA has less running time. \\

\noindent\textbf{Compute vertices with perturbation:}
In this section, we compare the robustness  of AVTA with multiple random projections presented in section~\ref{avtamultiple}  with Fast Anchor Word ~\citet{arora2013practical}. Instead of  actual set of points $S$ as input, the algorithm is given a perturbed set $S_\circ$, i.e. $S$ is corrupted by some noise.
Having fixed $K=100$,$n=500$, $m=100$ , we  choose a Gaussian perturbation from $\mathcal{N}(0,\tau) ^m$ where $\tau$ varies from $0.3$ to $3$.  In case of general convex hull, a failure of Fast Anchor Word on computing vertices of general convex hull is presented. The data is generated by setting $\tau=0.3$, $m=50$, $n=500$ and let $K$ varies from $10 \sim 100$.
We  do an error analysis and evaluate the output of the algorithms by measuring the $l_2$ distance between true vertices and the convex hull of output vertices of the two algorithms.  More precisely, given a true vertex $v_i \in S$ and $\hat{S}$, the output of an algorithm, the error in recovering $v_i$ is defined to be $\min\limits_{u\in conv(\hat{S})} \; ||u-v_i||_2.$ We add up all the errors to get the total accumulated error.

\noindent \textbf{Results on computing perturbed vertices:}

\noindent
The recovery error in robustness comparison is shown in Table~\ref{tb:avta rec}. The AVTA with multiple random projection has a better recovery error in the simplex case.

\noindent It can also be observed from Figure~\ref{fig:synthetic-error-general} that in general case, as number of vertices exceeds the number of dimensions, Fast Anchor Word fails to recover more vertices and its error explodes.  \\

\begin{table}[!h]
	
	\centering
	\caption{Recovery error (Simplex)}
	\label{tb:avta rec}
	\begin{tabular}{|l|l|l|l|l|l|}
		\hline
		Var & AVTA+Multiple Rp & Fast Anhor & variance & AVTA+Multiple Rp & Fast Anhor \\ \hline
		0.3      & 2.96             & 2.96       & 1.8      & 16.60            & 17.98      \\ \hline
		0.6      & 5.79             & 5.79       & 2.1      & 19.40            & 20.58      \\ \hline
		0.9      & 8.61             & 9.36       & 2.4      & 21.93            & 23.77      \\ \hline
		1.2      & 11.34            & 12.00      & 2.7      & 23.69            & 24.90      \\ \hline
		1.5      & 14.16            & 15.44      & 3        & 26.72            & 28.78      \\ \hline
	\end{tabular}
\end{table}

%
%

\begin{figure}[!h]
	\centering
	
	\begin{subfigure}{0.45\textwidth}
		\centering
		\includegraphics[width=0.95\textwidth]{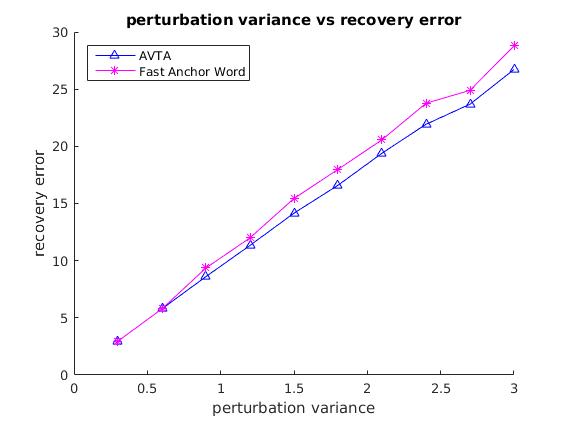}
		\caption{Recovery error-computing vertices(simplex case)}
		\label{fig:avta rec}
	\end{subfigure}%
	\begin{subfigure}{0.45\textwidth}
		\centering
		\includegraphics[width=0.95\textwidth]{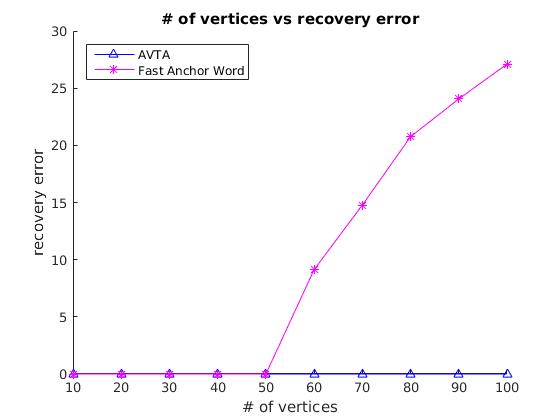}
		\caption{Recovery error-computing vertices(eneral convex hull)}
		\label{fig:synthetic-error-general}
	\end{subfigure}%
	
	\begin{subfigure}{0.45\textwidth}
		\centering
		\includegraphics[width=0.95\textwidth]{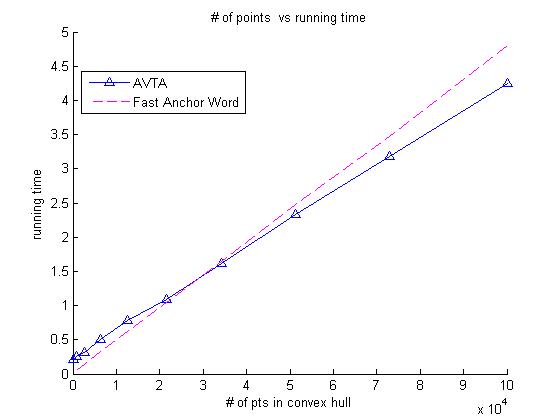}
		\caption{Running time (secs) of computing vertices of simplex}
		\label{fig:run_time_avta_faw}
	\end{subfigure}%
	\begin{subfigure}{0.45\textwidth}
		\centering
		\includegraphics[width=0.95\textwidth]{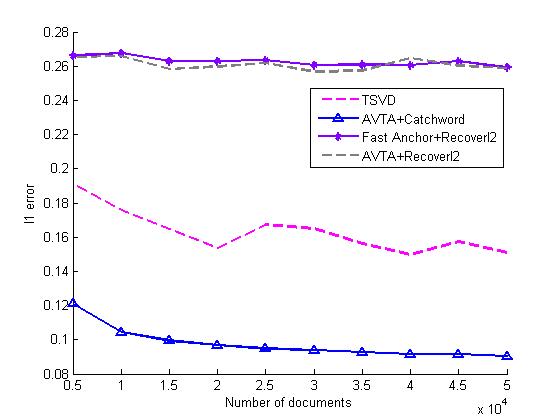}
		\caption{$\ell_1$ error in the semi-synthetic dataset.}
		\label{fig:semi-synthetic}
	\end{subfigure}%
	
	\begin{subfigure}{0.45\textwidth}
		\centering
		\includegraphics[width=0.95\textwidth]{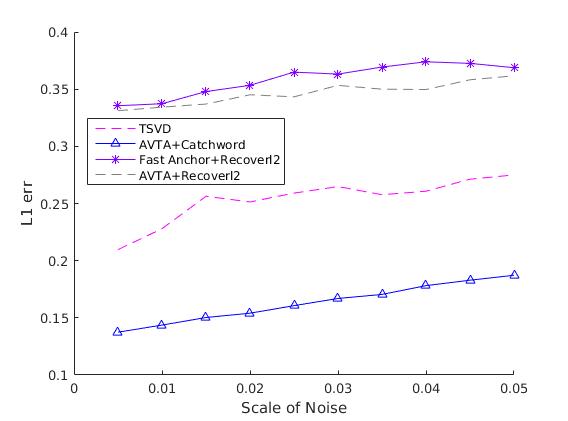}
		\caption{$\ell_1$ error in the perturbed semi-synthetic dataset.}
		\label{fig:semi-synthetic-noise}
	\end{subfigure}%
	\begin{subfigure}{0.45\textwidth}
		\centering
		\includegraphics[width=0.95\textwidth]{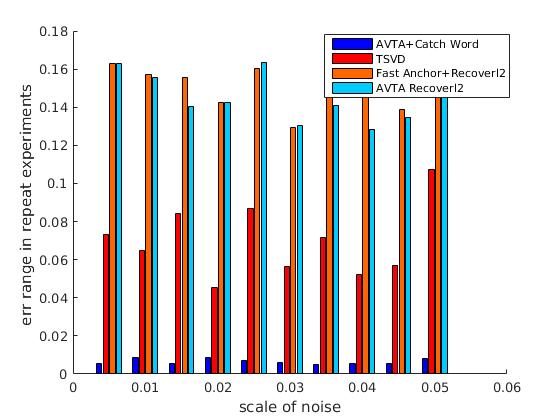}
		\caption{Range of the $\ell_1$ error over $10$ runs on the noisy semi-synthetic dataset.}
		\label{fig:err_range}
	\end{subfigure}%
	\caption{}
\end{figure}

\subsection{Topic modeling}

\noindent
We compare our algorithms with the Fast Anchor + Recoverl2 algorithm of~\citet{arora2013practical} and the TSVD algorithm of~\citet{bansal2014provable} on two types of data sets: semi-synthetic data and real world data.
We next describe our methodology and empirical results in detail.

\noindent \textbf{Semi Synthetic Data:} For Semi-Synthetic data set, we use similar methodology as in~\citet{arora2013practical}. We first train the model on real data set using Gibbs sampling with $1,000$ iterations. We choose $50$ as the number of topics which follows ~\citet{bansal2014provable}.  Given the parameters learned from dataset, we generate documents with $\alpha$ set to be $0.01$. The average document length is $1,000$.  Then the reconstruction error is measured by the $l_1$ distance of bipartite matched pairs between  the true word-topic distribution and the  word-topic distribution~\citet{arora2013practical}. We then average the errors to compute the final mean error.

\noindent \textbf{Real Data:} We use the \textbf{NIPS} data set with 1500 documents , and a pruned vocabulary of 2K words, and the NYTimes Corpus with sub sampled $30,000$ documents, and a pruned vocabulary of 5k words.~\footnote{\url{https://archive.ics.uci.edu/ml/datasets/bag+of+words}}. For the real world data set, as in prior works~\citet{arora2013practical, bansal2014provable}, we evaluate the coherence to measure topic quality~\citet{yao2009efficient}. Given a set of words $\mathcal{W}$ associated with a learned topic, the coherence is computed as:
$Coherence(\mathcal{W})= \sum_{w_1,w_2 \in \mathcal{W}}  \log\frac{D(w_1,w_2)+\epsilon}{D(w_2)}$,
where $D(w_1)$ and $D(w_1,w_2)$ are the number of documents where $w_1$ appears and $(w_1,w_2)$ appear together respectively~\citet{arora2013practical}, and $\varepsilon$ is set to $0.01$ to avoid $w_1,w_2$ that never co-occur~\citet{stevens2012exploring}. The total coherence is the sum of the coherence of each topic. In the NIPS dataset, $1,000$ out of the $1,500$ documents were selected as the training set to learn the word-topic distributions. The rest of the documents were used as the testing set.

\noindent \textbf{Implementation Details:} We compare 4 algorithms, {AVTA+CatchWord}, {TSVD}, the {Fast Anchor + Recoverl2} and the {AVTA+Recoverl2}. We implement our own version of Fast Anchor + Recoverl2 as described in~\citet{arora2013practical}. TSVD is implemented using the code provided by the authors in~\citet{bansal2014provable}. AVTA+Recoverl2 corresponds to using AVTA to detect anchor words from the word-word covariance matrix and then using the Recoverl2 procedure from~\citet{arora2013practical} to get the topic-word matrix. AVTA + CatchWord corresponds to finding the low dimensional embedding of each document in terms of the coefficient vector of its representation in the convex hull of the vertices. The next step is to cluster these points. In practice, one could use the Lloyd's algorithm for this step which could be sensitive to initialization.   To remedy this, we use similar heuristic as ~\citet{bansal2014provable} of the initialization step. We repeat AVTA for $3$ times and pick the set of vertices with highest quality where the quality is measured by sum of distances of each vertex to convex hull of other vertices. We set the number of output vertices $K=50$ which is the same as the number of topics. i.e. each vertex corresponds to a topic. We found that initializing by simply assigning clusters using neighborhoods of highest degree vertices works effectively. As a final step, we use the post processing step from~\citet{bansal2014provable} to recover the topic-word matrix from the clustering.

\noindent \textbf{Robustness:} We also generate perturbed version of the semi synthetic data. We generate a random matrix with i.i.d. entries uniformly distributed with different scales varying from $0.005-0.05$. We test all the algorithms with the document-word matrix added with the noise matrix.

\noindent \textbf{Results on Semi Synthetic Data:}
Figures~\ref{fig:semi-synthetic} and \ref{fig:semi-synthetic-noise} show the $\ell_1$ reconstruction of all the four algorithms under both clean and noisy versions of the semi synthetic data set. For topic $i$, let $A_i$ be the ground truth topic vector and $\hat{A}_i$ be the topic vector recovered by the algorithm. Then the $\ell_1$ error is defined as $\frac 1 K \sum_{i=1}^K \|A_i - \hat{A}_i\|_1$. The plots show that AVTA+CatchWord is consistently better than both TSVD and Fast Anchor + Recoverl2 and produces significantly more accurate topic vectors. In order to further test the robustness of our approach, we plot in Figure~\ref{fig:err_range} the range of the $\ell_1$ error obtained across multiple runs of the algorithms on the same data set. The range is defined to be the difference between the maximum and the minimum error recovered by the algorithm across different runs. We see that AVTA+CatchWord produces solutions that are much more stable to the effect of the noise as compared to other algorithms. Table~\ref{tb:runtime-4} shows the running time of the experiments of 4 algorithms. As can be seen, when using AVTA to learn topic models via the anchor words approach, our algorithm has comparable run time to Fast Anchor + Recoverl2.  In CatchWord based learning, computing vertices is expensive compared to K-SVD step of TSVD thus AVTA
has longer running time.


\begin{table}[h]
	\centering
	\caption{Running time  of algorithms on semi synthetic data (secs)}
	\label{tb:runtime-4}
	\scalebox{0.9}{
		\begin{tabular}{|l|l|l|l|l|}
			\hline
			\begin{tabular}[c]{@{}l@{}}Num of\\   documents\end{tabular} & 5,000 & 15,000 & 30,000 & 50,000 \\ \hline
			Fast anchor+Recoverl2                                        & 5.49    & 6.00     & 10.30    & 13.60    \\ \hline
			AVTA+Recoverl2                                               & 7.82    & 7.68     & 12.84    & 16.40    \\ \hline
			TSVD                                                         & 17.02   & 43.27    & 81.24    & 112.80   \\ \hline
			AVTA+Catch Word                                              & 29.89   & 120.04   & 372.17   & 864.30   \\ \hline
		\end{tabular}
	}
\end{table}

\noindent \textbf{Results on Real Data:}
Table ~\ref{tb:real}  shows the topic coherence obtained by the algorithms. One can see that in both the approaches, either via anchor words or the clustering approach, AVTA based algorithms perform comparably to state of the art methods ~\footnote{The topic coherence results for TSVD do not match the ones presented in~\citet{bansal2014provable} since in their experiments, the authors look at top 10 most frequent words in each topic. In our experiments we compute coherence for the top 5 most frequent words in each topic.}. The running time is presented in Table ~\ref{tb:running time}.  The AVTA+CathchWord has less running time in the real data experiments. Per our observation, the convex hull of  word-document vectors in real data set has more vertices than $K$, the number of topics. The AVTA catches $K$ vertices efficiently due to its small number of iterations on line search for $\gamma$. In semi-synthetic data set, the number of 'robust' vertices is approximately the same as number of topics $K$ thus AVTA needs to find almost all vertices. To catch enough vertices, AVTA needs several iterations decreasing $\gamma$ which is computationally expensive.

\subsection{Non-negative matrix factorization}
\noindent \textbf{AVTA for NMF:}
For our experiments on NMF we use the Swimmer data set~\citet{donoho2003does} that consists of $256$ swimmer figures with each a $32 \times 32$ binary pixel images. One can interpret each image as a document and pixels as a word in the document~\citet{ding2013topic}. All swimmers consist of $4$ limbs with each limb having $4$ different possible poses. One can then consider the different poses of limbs as the true underlying topics~\citet{donoho2003does}. We compare the algorithm proposed in~\citet{arora2012computing} with {AVTA+Recoverl2} on the swimmer data set. We  construct a noisy version by adding spurious poses to original swimmer data set. Let $\Omega(A)$ be a function that outputs a randomly chosen $32 \times 8$ block of an image. We generate a  'spurious pose' of size  $32\times 8$ by $\Omega(M_i)$ where $M_i$ is a randomly chosen swimmer image. Then we take another randomly chosen image $M_j$ and compute the corrupted image as $M'_j= M_j + c \cdot \Omega(M_i)$ where we simply set $c=0.1$. An illustration of the noise data set is shown in Figure~\ref{fig:Swimmer}. Since the true underlying topics are known, we will plot the output of the algorithms and compare it with the underlying truth.

\noindent \textbf{Results on NMF:}  We compare the performance of AVTA on these data sets with the performance of the Separable NMF algorithm proposed in~\citet{arora2012computing}. Figures~\ref{fig:NMF swim} and ~\ref{fig:AVTA swim} show the output of the Separable NMF algorithm and that of our algorithm respectively on the noisy data set. Our approach produces competitive results as compared to the Separable NMF algorithm.

\begin{table}[h]
	\centering
	\caption{Topic coherence on real data}
	\label{tb:real}
	\scalebox{0.7}{
		\begin{tabular}{|l|l|l|l|l|}
			\hline
			& \tiny {Fast Anchor+RecoverL2} & \tiny { AVTA+RecoverL2} & TSVD     &  \tiny {AVTA+Catch Word} \\ \hline
			NIPS    & -15.8 $\pm 2.24$  & -16.04 $\pm 2.09$         & -16.86 $\pm 1.66$   & -18.65  $\pm 1.78$         \\ \hline
			NYTimes  & -32.15  $\pm 2.7$                & -32.13    $\pm 2.43$     & -29.39 $\pm 1.43$ & -30.13     $\pm 1.98$       \\ \hline
		\end{tabular}
	}
\end{table}

\begin{table}[!h]
	\centering
	\caption{Running time on real data experiments (secs)}
	\label{tb:running time}
	\scalebox{0.7}{
		\begin{tabular}{|l|l|l|l|l|}
			\hline
			&  \tiny {Fast Anchor+RecoverL2} &  \tiny {AVTA+RecoverL2} & TSVD  &  \tiny {AVTA+Catch Word} \\ \hline
		
			NIPS    & 3.22                  & 4.41           & 56.58 & 22.78           \\ \hline
			NYTimes & 26.05                 & 27.79          & 237.6 & 101.07          \\ \hline
		\end{tabular}
	}
\end{table}

\begin{figure}[!h]
	\centering
	\scalebox{0.55}{
		\begin{subfigure}{0.65\textwidth}
			\centering
			\includegraphics[width=1.1\textwidth]{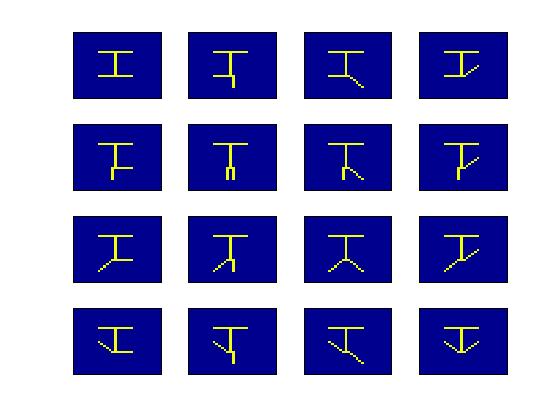}
			\caption{An example of  swimmer images.}
			\label{fig:Swimmer}
		\end{subfigure}
		\begin{subfigure}{0.65\textwidth}
			\centering
			\includegraphics[width=1.1\textwidth]{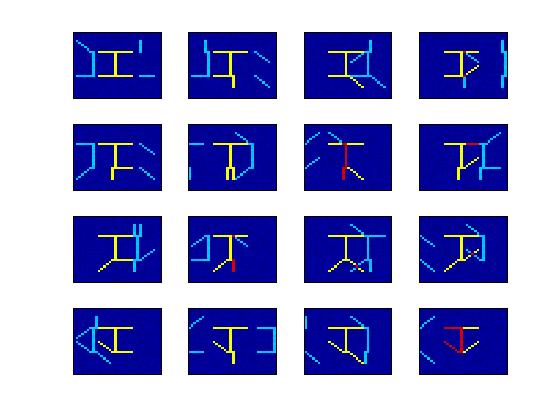}
			\caption{An example of spurious actions in swimmer images.}
			\label{fig:Swimmer}
		\end{subfigure}
	}
	\scalebox{0.55}{
		\begin{subfigure}{0.65\textwidth}
			\centering
			\includegraphics[width=1.1\textwidth]{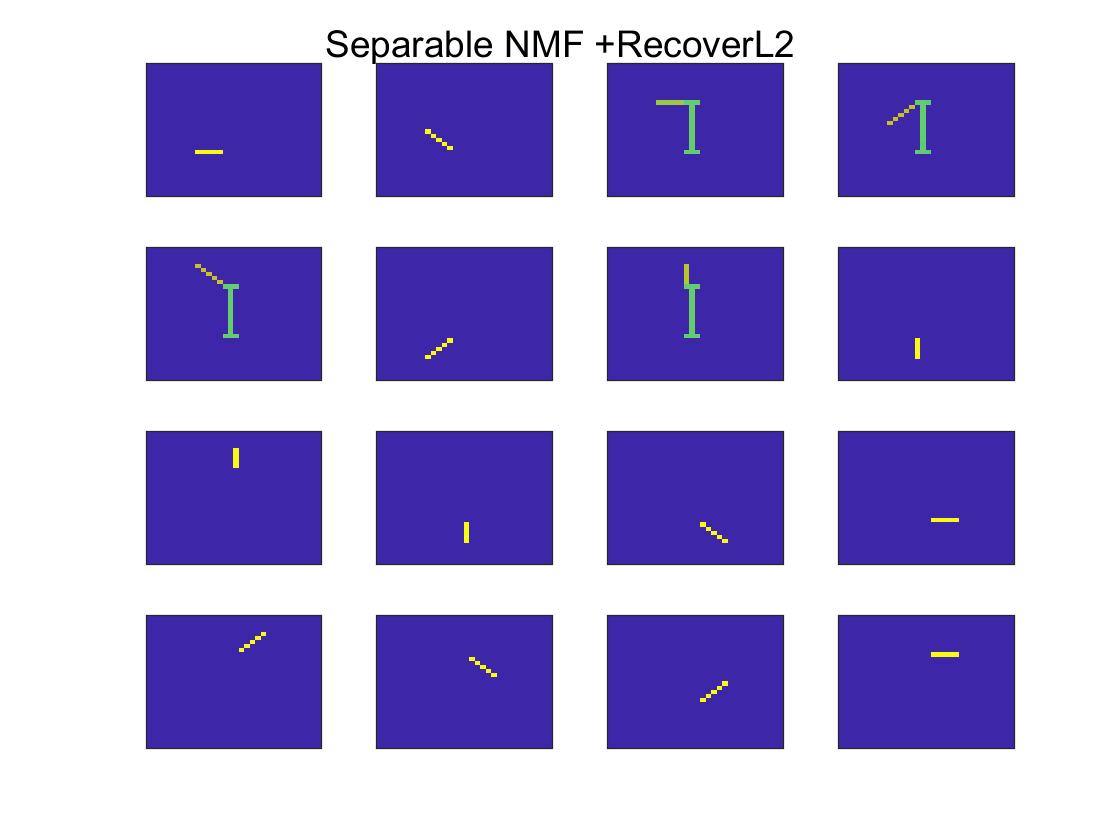}
			\caption{Output of {NMF +RecoverL2}}
			\label{fig:NMF swim}
		\end{subfigure}%
		\begin{subfigure}{0.65\textwidth}
			\centering
			\includegraphics[width=1.1\textwidth]{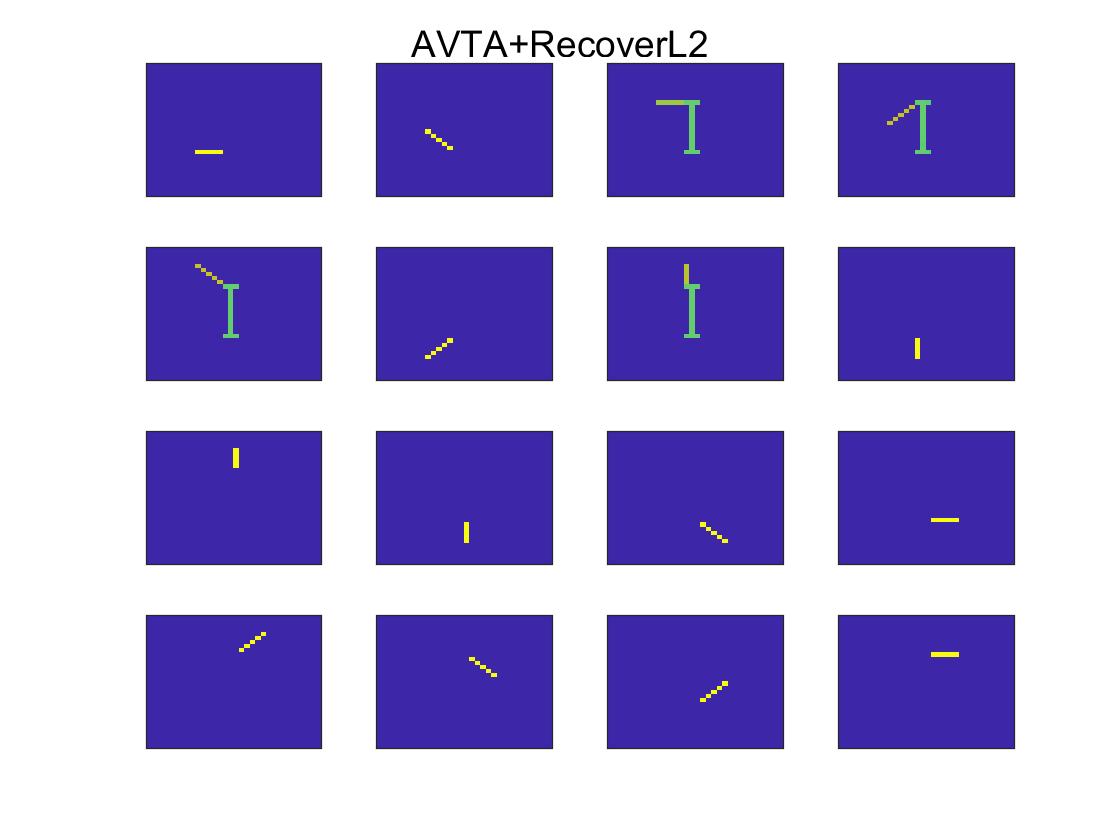}
			\caption{Output of {AVTA +RecoverL2}}
			\label{fig:AVTA swim}
		\end{subfigure}
	}
	\caption{}
	
\end{figure}

\section{Conclusion}
In this work we have presented a fast and robust algorithm for computing the vertices of the convex hull of
a set of points. Our algorithm efficiently computes the vertices of convex hulls in high dimensions and
even in the special case of the simplex is competitive with the state of the art approaches in terms of running time~\citet{arora2013practical}. Furthermore, our algorithm leads to an improved algorithm for topic modeling that is more robust and produces better approximations to the topic-word matrix. It will be interesting to provide theoretical claims supporting this observation in the context of specific applications. Furthermore, we believe that our algorithm will have more applications in machine learning problems beyond the ones
investigated here as well as applications in computational geometry and in linear programming.

\nocite{*}
\bibliographystyle{plainnat}
\bibliography{nips_2017}

\begin{thebibliography}{42}
\providecommand{\natexlab}[1]{#1}
\providecommand{\url}[1]{\texttt{#1}}
\expandafter\ifx\csname urlstyle\endcsname\relax
  \providecommand{\doi}[1]{doi: #1}\else
  \providecommand{\doi}{doi: \begingroup \urlstyle{rm}\Url}\fi

\bibitem[Anandkumar et~al.(2012)Anandkumar, Foster, Hsu, Kakade, and
  Liu]{anandkumar2012spectral}
Anima Anandkumar, Dean~P Foster, Daniel~J Hsu, Sham~M Kakade, and Yi-Kai Liu.
\newblock A spectral algorithm for latent dirichlet allocation.
\newblock In \emph{Advances in Neural Information Processing Systems}, pages
  917--925, 2012.

\bibitem[Arora et~al.(2012{\natexlab{a}})Arora, Ge, Kannan, and
  Moitra]{arora2012computing}
Sanjeev Arora, Rong Ge, Ravindran Kannan, and Ankur Moitra.
\newblock Computing a nonnegative matrix factorization--provably.
\newblock In \emph{Proceedings of the forty-fourth annual ACM symposium on
  Theory of computing}, pages 145--162. ACM, 2012{\natexlab{a}}.

\bibitem[Arora et~al.(2012{\natexlab{b}})Arora, Ge, and
  Moitra]{arora2012learning}
Sanjeev Arora, Rong Ge, and Ankur Moitra.
\newblock Learning topic models--going beyond svd.
\newblock In \emph{Foundations of Computer Science (FOCS), 2012 IEEE 53rd
  Annual Symposium on}, pages 1--10. IEEE, 2012{\natexlab{b}}.

\bibitem[Arora et~al.(2013)Arora, Ge, Halpern, Mimno, Moitra, Sontag, Wu, and
  Zhu]{arora2013practical}
Sanjeev Arora, Rong Ge, Yonatan Halpern, David Mimno, Ankur Moitra, David
  Sontag, Yichen Wu, and Michael Zhu.
\newblock A practical algorithm for topic modeling with provable guarantees.
\newblock In \emph{International Conference on Machine Learning}, pages
  280--288, 2013.

\bibitem[Arora et~al.(2014)Arora, Ge, and Moitra]{arora2014new}
Sanjeev Arora, Rong Ge, and Ankur Moitra.
\newblock New algorithms for learning incoherent and overcomplete dictionaries.
\newblock In \emph{COLT}, pages 779--806, 2014.

\bibitem[Awasthi and Risteski(2015)]{awasthi2015some}
Pranjal Awasthi and Andrej Risteski.
\newblock On some provably correct cases of variational inference for topic
  models.
\newblock In \emph{Advances in Neural Information Processing Systems}, pages
  2098--2106, 2015.

\bibitem[Bansal et~al.(2014)Bansal, Bhattacharyya, and
  Kannan]{bansal2014provable}
Trapit Bansal, Chiranjib Bhattacharyya, and Ravindran Kannan.
\newblock A provable svd-based algorithm for learning topics in dominant
  admixture corpus.
\newblock In \emph{Advances in Neural Information Processing Systems}, pages
  1997--2005, 2014.

\bibitem[Barber et~al.(1996)Barber, Dobkin, and Huhdanpaa]{barber1996quickhull}
C~Bradford Barber, David~P Dobkin, and Hannu Huhdanpaa.
\newblock The quickhull algorithm for convex hulls.
\newblock \emph{ACM Transactions on Mathematical Software (TOMS)}, 22\penalty0
  (4):\penalty0 469--483, 1996.

\bibitem[Blei(2012)]{blei2012probabilistic}
David~M Blei.
\newblock Probabilistic topic models.
\newblock \emph{Communications of the ACM}, 55\penalty0 (4):\penalty0 77--84,
  2012.

\bibitem[Blei et~al.(2003)Blei, Ng, and Jordan]{blei2003latent}
David~M Blei, Andrew~Y Ng, and Michael~I Jordan.
\newblock Latent dirichlet allocation.
\newblock \emph{Journal of machine Learning research}, 3\penalty0
  (Jan):\penalty0 993--1022, 2003.

\bibitem[Blum et~al.(2016)Blum, Har-Peled, and Raichel]{blum2016sparse}
Avrim Blum, Sariel Har-Peled, and Benjamin Raichel.
\newblock Sparse approximation via generating point sets.
\newblock In \emph{Proceedings of the twenty-seventh annual ACM-SIAM symposium
  on Discrete algorithms}, pages 548--557. Society for Industrial and Applied
  Mathematics, 2016.

\bibitem[Burges(1998)]{burges1998tutorial}
Christopher~JC Burges.
\newblock A tutorial on support vector machines for pattern recognition.
\newblock \emph{Data mining and knowledge discovery}, 2\penalty0 (2):\penalty0
  121--167, 1998.

\bibitem[Chan(1996{\natexlab{a}})]{chan1996optimal}
Timothy~M Chan.
\newblock Optimal output-sensitive convex hull algorithms in two and three
  dimensions.
\newblock \emph{Discrete \& Computational Geometry}, 16\penalty0 (4):\penalty0
  361--368, 1996{\natexlab{a}}.

\bibitem[Chan(1996{\natexlab{b}})]{chan1996output}
Timothy~M Chan.
\newblock Output-sensitive results on convex hulls, extreme points, and related
  problems.
\newblock \emph{Discrete \& Computational Geometry}, 16\penalty0 (4):\penalty0
  369--387, 1996{\natexlab{b}}.

\bibitem[Chazelle(1993)]{chazelle1993optimal}
Bernard Chazelle.
\newblock An optimal convex hull algorithm in any fixed dimension.
\newblock \emph{Discrete \& Computational Geometry}, 10\penalty0 (1):\penalty0
  377--409, 1993.

\bibitem[Chvatal(1983)]{chvatal1983linear}
Vasek Chvatal.
\newblock \emph{Linear programming}.
\newblock Macmillan, 1983.

\bibitem[Clarkson(1994)]{clarkson1994more}
Kenneth~L Clarkson.
\newblock More output-sensitive geometric algorithms.
\newblock In \emph{Foundations of Computer Science, 1994 Proceedings., 35th
  Annual Symposium on}, pages 695--702. IEEE, 1994.

\bibitem[Clarkson(2010)]{clarkson2010coresets}
Kenneth~L Clarkson.
\newblock Coresets, sparse greedy approximation, and the frank-wolfe algorithm.
\newblock \emph{ACM Transactions on Algorithms (TALG)}, 6\penalty0
  (4):\penalty0 63, 2010.

\bibitem[Ding et~al.(2013)Ding, Rohban, Ishwar, and Saligrama]{ding2013topic}
Weicong Ding, Mohammad~Hossein Rohban, Prakash Ishwar, and Venkatesh Saligrama.
\newblock Topic discovery through data dependent and random projections.
\newblock In \emph{ICML (3)}, pages 1202--1210, 2013.

\bibitem[Donoho and Stodden(2003)]{donoho2003does}
David Donoho and Victoria Stodden.
\newblock When does non-negative matrix factorization give a correct
  decomposition into parts?
\newblock In \emph{Advances in Neural Information Processing Systems}, 2003.

\bibitem[Frank and Wolfe(1956)]{frank1956algorithm}
Marguerite Frank and Philip Wolfe.
\newblock An algorithm for quadratic programming.
\newblock \emph{Naval Research Logistics (NRL)}, 3\penalty0 (1-2):\penalty0
  95--110, 1956.

\bibitem[G{\"a}rtner and Jaggi(2009)]{gartner2009coresets}
Bernd G{\"a}rtner and Martin Jaggi.
\newblock Coresets for polytope distance.
\newblock In \emph{Proceedings of the twenty-fifth annual symposium on
  Computational geometry}, pages 33--42. ACM, 2009.

\bibitem[Gilbert(1966)]{gilbert1966iterative}
Elmer~G Gilbert.
\newblock An iterative procedure for computing the minimum of a quadratic form
  on a convex set.
\newblock \emph{SIAM Journal on Control}, 4\penalty0 (1):\penalty0 61--80,
  1966.

\bibitem[Har-Peled et~al.(2007)Har-Peled, Roth, and Zimak]{har2007maximum}
Sariel Har-Peled, Dan Roth, and Dav Zimak.
\newblock Maximum margin coresets for active and noise tolerant learning.
\newblock In \emph{IJCAI}, pages 836--841, 2007.

\bibitem[Jaggi(2013)]{jaggi2013revisiting}
Martin Jaggi.
\newblock Revisiting frank-wolfe: Projection-free sparse convex optimization.
\newblock 2013.

\bibitem[Jarvis(1973)]{jarvis1973identification}
Ray~A Jarvis.
\newblock On the identification of the convex hull of a finite set of points in
  the plane.
\newblock \emph{Information Processing Letters}, 2\penalty0 (1):\penalty0
  18--21, 1973.

\bibitem[Jin and Kalantari(2006)]{jin2006procedure}
Yi~Jin and Bahman Kalantari.
\newblock A procedure of chv{\'a}tal for testing feasibility in linear
  programming and matrix scaling.
\newblock \emph{Linear algebra and its applications}, 416\penalty0
  (2-3):\penalty0 795--798, 2006.

\bibitem[Johnson and Lindenstrauss(1984)]{johnson1984extensions}
William~B Johnson and Joram Lindenstrauss.
\newblock Extensions of lipschitz mappings into a hilbert space.
\newblock \emph{Contemporary mathematics}, 26\penalty0 (189-206):\penalty0 1,
  1984.

\bibitem[Kalantari(2015)]{kalantari2015characterization}
Bahman Kalantari.
\newblock A characterization theorem and an algorithm for a convex hull
  problem.
\newblock \emph{Annals of Operations Research}, 226\penalty0 (1):\penalty0
  301--349, 2015.

\bibitem[Karmarkar(1984)]{karmarkar1984new}
Narendra Karmarkar.
\newblock A new polynomial-time algorithm for linear programming.
\newblock In \emph{Proceedings of the sixteenth annual ACM symposium on Theory
  of computing}, pages 302--311. ACM, 1984.

\bibitem[Khachiyan(1980)]{khachiyan1980polynomial}
Leonid~G Khachiyan.
\newblock Polynomial algorithms in linear programming.
\newblock \emph{USSR Computational Mathematics and Mathematical Physics},
  20\penalty0 (1):\penalty0 53--72, 1980.

\bibitem[Lee and Seung(2001)]{lee2001algorithms}
Daniel~D Lee and H~Sebastian Seung.
\newblock Algorithms for non-negative matrix factorization.
\newblock In \emph{Advances in neural information processing systems}, pages
  556--562, 2001.

\bibitem[Matou{\v{s}}ek and Schwarzkopf(1992)]{matouvsek1992linear}
Ji{\v{r}}{\'\i} Matou{\v{s}}ek and Otfried Schwarzkopf.
\newblock Linear optimization queries.
\newblock In \emph{Proceedings of the eighth annual symposium on Computational
  geometry}, pages 16--25. ACM, 1992.

\bibitem[Olshausen and Field(1996)]{olshausen1996emergence}
Bruno~A Olshausen and David~J Field.
\newblock Emergence of simple-cell receptive field properties by learning a
  sparse code for natural images.
\newblock \emph{Nature}, 381\penalty0 (6583):\penalty0 607, 1996.

\bibitem[Papadimitriou et~al.(1998)Papadimitriou, Tamaki, Raghavan, and
  Vempala]{papadimitriou1998latent}
Christos~H Papadimitriou, Hisao Tamaki, Prabhakar Raghavan, and Santosh
  Vempala.
\newblock Latent semantic indexing: A probabilistic analysis.
\newblock In \emph{Proceedings of the seventeenth ACM SIGACT-SIGMOD-SIGART
  symposium on Principles of database systems}, pages 159--168. ACM, 1998.

\bibitem[Spielman et~al.(2012)Spielman, Wang, and Wright]{spielman2012exact}
Daniel~A Spielman, Huan Wang, and John Wright.
\newblock Exact recovery of sparsely-used dictionaries.
\newblock In \emph{COLT}, pages 37--1, 2012.

\bibitem[Stevens et~al.(2012)Stevens, Kegelmeyer, Andrzejewski, and
  Buttler]{stevens2012exploring}
Keith Stevens, Philip Kegelmeyer, David Andrzejewski, and David Buttler.
\newblock Exploring topic coherence over many models and many topics.
\newblock In \emph{Proceedings of the 2012 Joint Conference on Empirical
  Methods in Natural Language Processing and Computational Natural Language
  Learning}, pages 952--961. Association for Computational Linguistics, 2012.

\bibitem[Toth et~al.(2004)Toth, O'Rourke, and Goodman]{toth2004handbook}
Csaba~D Toth, Joseph O'Rourke, and Jacob~E Goodman.
\newblock \emph{Handbook of discrete and computational geometry}.
\newblock CRC press, 2004.

\bibitem[Toussaint(1983)]{toussaint1983solving}
Godfried~T Toussaint.
\newblock Solving geometric problems with the rotating calipers.
\newblock In \emph{Proc. IEEE Melecon}, volume~83, page A10, 1983.

\bibitem[Vu et~al.(2017)Vu, Poirion, and Liberti]{vu2017random}
Ky~Vu, Pierre-Louis Poirion, and Leo Liberti.
\newblock Random projections for linear programming.
\newblock \emph{arXiv preprint arXiv:1706.02768}, 2017.

\bibitem[Yao et~al.(2009)Yao, Mimno, and McCallum]{yao2009efficient}
Limin Yao, David Mimno, and Andrew McCallum.
\newblock Efficient methods for topic model inference on streaming document
  collections.
\newblock In \emph{Proceedings of the 15th ACM SIGKDD international conference
  on Knowledge discovery and data mining}, pages 937--946. ACM, 2009.

\bibitem[Zhang(2003)]{zhang2003sequential}
Tong Zhang.
\newblock Sequential greedy approximation for certain convex optimization
  problems.
\newblock \emph{IEEE Transactions on Information Theory}, 49\penalty0
  (3):\penalty0 682--691, 2003.

\end{thebibliography}

\end{document}